\documentclass[10pt]{amsart}

\usepackage{macros,slashed}

\usepackage{parskip}
\title{The holomorphic bosonic string}

\usepackage{tikz}
\usetikzlibrary{arrows,shapes}
\usetikzlibrary{trees}
\usetikzlibrary{matrix,arrows}
\usetikzlibrary{positioning}
\usetikzlibrary{calc,through}
\usetikzlibrary{decorations.pathreplacing}
\usepackage{pgffor}

\usetikzlibrary{decorations.pathmorphing}
\usetikzlibrary{decorations.markings}
\tikzset{
    vector/.style={decorate, decoration={snake}, draw},
	provector/.style={decorate, decoration={snake,amplitude=2.5pt}, draw},
	antivector/.style={decorate, decoration={snake,amplitude=-2.5pt}, draw},
    fermion/.style={draw=black, postaction={decorate},
        decoration={markings,mark=at position .55 with {\arrow[draw=black]{>}}}},
    fermionbar/.style={draw=black, postaction={decorate},
        decoration={markings,mark=at position .55 with {\arrow[draw=black]{<}}}},
    fermionnoarrow/.style={draw=black},
    gluon/.style={decorate, draw=black,
        decoration={coil,amplitude=4pt, segment length=5pt}},
    scalar/.style={dashed,draw=black, postaction={decorate},
        decoration={markings,mark=at position .55 with {\arrow[draw=black]{>}}}},
    scalarbar/.style={dashed,draw=black, postaction={decorate},
        dwecoration={markings,mark=at position .55 with {\arrow[draw=black]{<}}}},
    scalarnoarrow/.style={dashed,draw=black},
    electron/.style={draw=black, postaction={decorate},
        decoration={markings,mark=at position .55 with {\arrow[draw=black]{>}}}},
	bigvector/.style={decorate, decoration={snake,amplitude=4pt}, draw},
}

\begin{document}

\author{Owen Gwilliam}
\address{Max Planck Institute for Mathematics, Bonn, Germany}
\email{gwilliam@mpim-bonn.mpg.de}

\author{Brian Williams}
\address{Department of Mathematics, Northwestern University, Evanston, IL}
\email{brianwilliams2012@u.northwestern.edu}


\date{}

\begin{abstract}
We describe and analyze a holomorphic version of the bosonic string in the formalism of quantum field theory developed by Costello and collaborators, which provides a powerful combination of renormalization theory and the Batalin-Vilkovisky formalism. Our focus here is on the case in which the target space-time is a vector space. We identify the critical dimension as an obstruction to satisfying the quantum master equation, and when the obstruction vanishes, we construct a one-loop exact quantization. Moreover, we show how the factorization algebra of observables recovers the BRST cohomology of the string and use this perspective to give a new construction of its Gerstenhaber structure. Finally, we show how the factorization homology along closed manifolds encodes the determinant line bundle over the moduli space of Riemann surfaces.
An auxiliary goal of this paper is to give an exposition of this QFT formalism with the holomorphic bosonic string theory as the running example.
\end{abstract}

\maketitle

\setcounter{tocdepth}{1}

\tableofcontents

\section{Introduction}

Two intertwined goals govern our exposition.
First, we want to describe a two-dimensional field theory,
which we view as a holomorphic version of bosonic string theory,
and its perturbative quantization.
We'll see that this theory encodes the chiral sector of a bosonic string with linear target space,
justifying our interpretation.
Second, we want to use this theory as the running example for key ideas and techniques in the formalism for quantum field theory developed by Costello and collaborators \cite{CosBook, CG1,CG2, LL1, GG1, GLL, LiVA}.
We hope to give readers a feel for how to use this formalism by exhibiting it with a beautiful theory.

Our focus is thus on narrative rather than detailed argumentation.
That is, we work systematically according the natural flow of the formalism. 
Along the way we emphasize the motivations behind each step rather than the nitty-gritty computations. 
Precedence is given to communicating the essence of an argument, over spelling everything out.
We do give detailed citations where such arguments can be found in the literature,
but we defer some not-yet-extant details to a forthcoming work on this theory with curved target space~\cite{GWcurved}.

None of the results here about string theory is new, 
as the bosonic string has been under intensive study for several decades,
but this formalism recovers them in a single, systematic process,
often giving a novel argument or perspective.
It is compelling to have a direct path from the action functional to such sophisticated constructions as the semi-infinite cohomology of a vertex algebra.
In fact, since so many of these results are familiar,
the reader may see more clearly what's distinctive and illuminating about this approach to field theory.

There are many references on the bosonic string that have influenced us.
In the physics literature there are the classic sources \cite{GSW1, GSW2, polchinski} that explain perturbative string theory. 
In addition, there is an extensive mathematically-oriented treatment of perturbative string theory in \cite{DP}, as well as D'Hoker's notes in Volume II of \cite{IAScourse}.
Our approach, while intimately related, starts with a ``first-order" description of the bosonic string. 

Given the vastness of the string theory literature,
it should not be a surprise that there is already work along these lines,
notably by Losev, Marshakov, and Zeitlin \cite{LMZ}.
One could view this paper as attempting to communicate many of their insights to those with an intuition growing out of homotopical algebra and the functorial approach to geometry.
Again, we note that the formalism of Costello provides a mathematical articulation and verification of many ideas long known to physicists, such as the Wilsonian view of renormalization and the Batalin-Vilkovisky (BV) approach to gauge and gravity theories.\footnote{We also note that given the literature's size,
and our relative and unfortunate ignorance of much of it,
we have chosen to mention a reference when we feel its description is particularly useful for us, 
even if it is not the original or standard reference for the result.}
This machinery allows us to revisit such prior work in a manner particularly amenable to mathematicians.

\subsection{Overview} \label{sec:bvoverview}

The central figure of this paper is a holomorphic analogue of the bosonic string.
We proceed, as usual in physics, from the classical to the quantum.

Hence, we begin by introducing the classical theory, 
expressed both in the BV formalism and also in terms of an action functional.
We take some time to identify this theory as the chiral sector of a limit of the bosonic string,
where the K\"{a}hler metric of the target is made very large. 
We also interpret the theory in the language of derived geometry.

We then turn to analyzing the deformations of this classical theory,
which by Costello's work admits a nice description in terms of a type of Gelfand-Fuks cohomology.
This perspective naturally leads to a discussion of string backgrounds.

With a firm grip on the classical theory, we turn to constructing the perturbative quantization.
We first work with a disk or $\CC$ as the source manifold,
and we review relevant features of Costello's approach to renormalization.
The usual dimensional Weyl anomaly appears as an obstruction to satisfying the quantum master equation,
a key condition in the BV formalism.
At this stage, the anomaly appears as a computation with Feynman diagrams.

The next section describes the vertex algebra of the quantized theory,
using the machinery of factorization algebras of \cite{CG1, CG2}.
We find this piece of the formalism particularly illuminating,
as it lets a mathematician understand how to read off the OPE from path integral manipulations.

We then turn to the case of a compact Riemann surface as the source manifold.
Here we discuss how the formalism relates to the global approach to computing anomalies using, for instance, the Grothendieck-Riemann-Roch formula.
We also discuss conformal blocks in this formalism.

Finally, we sketch how to modify the approach here to allow a complex manifold as the target.
This paper can be viewed as an expository precursor to future work,
which pushes into new territory (particularly in describing the vertex algebra).

\subsection{Lessons to bear in mind}

Before turning to our example,
we want to expound some key ideas of the Costello formalism so that the reader is alert to them when proceeding through the text.
That is, we wish to articulate here the structural features of this BV/renormalization package that make the arguments below conceptual.

For instance, in a gauge theory we know that connections provide the ``naive'' fields and that one must identify connections that are gauge-equivalent.
A mathematician would say the true fields are a {\em stacky} quotient of the naive fields.
Similarly, the critical locus of the action functional $S$ is the zero locus of its differential $\d S$ (ignoring some subtleties of the variational set-up),
which is the intersection of $\d S$ with the zero section of the cotangent space of the fields.
But in mathematics it is better to take {\em derived} intersections.

\begin{lesson}[Part 1, \cite{CG2}]
The classical BV formalism is a method for computing the derived critical locus of the action functional on the derived stack of fields.
Ghosts appear to describe the direction along which one quotients---the stacky direction---while the antifields appear to describe the direction along which one intersects---the derived direction.
\end{lesson}

We will describe our theory in the usual way, involving fields and ghosts, 
but we will also sketch its meaning in terms of global derived geometry,
which we find illuminates the deep connections between string theory and algebraic geometry.

Path integral quantization amounts to trying to put a kind of measure or volume form on the derived stack of fields.
When the fields form a linear space, 
there is a natural quantization that is translation-invariant along the fields,
which is the analogue of the Lebesgue measure on an ordinary vector space.

\begin{lesson}[\cite{GH}]
Linear BV quantization is functorial, and it behaves much like a determinant functor.
Hence, when one takes the fiberwise quantization of a family of linear theories,
one typically obtains a determinant line bundle over the base.
\end{lesson}

This situation is relevant to us because the theory we study arises from a simple free theory,
the free $\beta\gamma$ system, which lives on any Riemann surface.
Hence the quantization of the free $\beta\gamma$ system makes sense over the moduli of Riemann surfaces and naturally produces a line bundle.

To be more specific, our classical theory of interest arises by gauging the natural action of holomorphic vector fields on the free $\beta\gamma$ system.
As holomorphic vector fields are infinitesimal biholomorphisms, 
one can say that we couple the $\beta\gamma$ system to holomorphic gravity.
But then we recognize a natural consequence of our prior lessons.

\begin{lesson}[\S 5.11, \cite{CosBook}]
Gauging a classical theory corresponds to taking a stacky quotient of the original fields. 
To quantize the gauged theory corresponds to descending the quantization to the quotient.
Hence, an anomaly that prevents quantization should be understood as an obstruction to descent.
\end{lesson}

The formalism of Costello makes this relationship manifest, 
as the anomaly that appears in trying to produce a BV quantization---which is a Feynman diagram construction---is a cocycle in a dg Lie algebra determined by the classical field theory.
Thus, the anomaly determines an element of a natural Lie algebra cohomology group (in this case, Gelfand-Fuks cohomology),
whose descent-theoretic meaning is typically easy to recognize. 
Here we will discover the famed Weyl, or conformal, anomaly, which requires the target space to be real 26-dimensional. 

Anomalies are often characteristic classes, and this BV/renormalization package offers a structural explanation.
Most classical field theories---at least most of broad interest---make sense on a class of manifolds,
and so the anomaly ought to be determined by the local geometry of this class.
In more mathematical language we have the following.

\begin{lesson}[\cite{GGW}]
If a classical theory determines a sheaf on some site of manifolds (such as the site of Riemann surfaces and local biholomorphisms), 
then to quantize the theory over the whole site, 
it suffices to check on a generating cover (typically given by disks with geometric structure) but compatibly with all automorphisms.
\end{lesson}

In particular, the BV anomaly is a cocycle for the Lie algebra of automorphisms of the {\em formal} disk equipped with such geometric structures.
In other words, it lives in some kind of Gelfand-Fuks cohomology, which gives deep and informative connections with foliation theory and topology.

So far, everything we have mentioned is well-known in field theory, 
albeit often expressed in a different dialect of mathematics.
We now turn to the main new notion of this framework:
factorization algebras, which provide an efficient and powerful way to organize the local-to-global structure of the observables of a field theory.

\begin{lesson}[\cite{CG1,CG2}]
Every BV theory produces a factorization algebra. 
The local structure encodes the OPE algebra, so that for a chiral CFT, one recovers a vertex algebra. 
On compact manifolds, the global structure often has finite-dimensional cohomology because solutions to the equations of motion are typically finite-dimensional.
For a chiral CFT, one recovers the conformal blocks in this way.
\end{lesson}

A technical result of \cite{CG1} gives a precise articulation of this lesson,
and we will apply it to identify the vertex algebra arising from our holomorphic version of the bosonic string.

\subsection{Acknowledgements}

We learned this approach to perturbative field theory as students of Kevin Costello.
He guided us towards this theory of the holomorphic string,  
and he pointed out key results and features visible through this BV/renormalization formalism. 
OG spent some time on this theory in graduate school,
partly in collaboration with Yuan Shen,
whom he thanks for illuminating discussions and computations.
The authors also wish to thank Si Li for his typical incisive comments and insight on CFT,
which clarified some of the trickier technical aspects.
We thank the referee for their careful reading of an earlier version of this manuscript, and for their revision requests that improved the overall exposition. 
BW acknowledges support of IH\'{E}S during his visit in 2017 and
funding from the European Research Council (ERC) under the European
Union's Horizon 2020 research and innovation program (QUASIFT grant
agreement 677368).
Finally, we must express deep gratitude toward the Max Planck Institute for Mathematics,
which supported OG throughout his work on this project and which also allowed BW to visit repeatedly and thus contributed substantially to the efficacy of this collaboration.

\section{The classical holomorphic bosonic string} \label{sec:classical}

There is a basic format for a string theory, at least in the perturbative approach. 
One starts with a nonlinear $\sigma$-model, 
whose fields are smooth maps from a Riemann surface to a target manifold $X$;
in this setting we want the theory to make sense for an arbitrary Riemann surface as the source manifold.
In the usual bosonic string theory, 
this nonlinear $\sigma$-model picks out the harmonic maps from a Riemannian 2-manifold to a Riemannian manifold.
In our holomorphic setting,
the nonlinear $\sigma$-model picks out holomorphic maps from a Riemann surface to a complex manifold.
One then quotients the space of fields (and solutions to the equations of motion) with respect to reparametrization.
In the usual bosonic string,
one quotients by diffeomorphisms and Weyl rescalings (i.e., multiplying the metric by a positive real function), 
which can thus change the metric on the source.
In our setting, we quotient by biholomorphisms, which act on the complex structure on the source.

In this section we begin by describing our theory in the BV formalism.
We do not expect the reader to find the action functional immediately clear,
so we devote some time to analyzing what it means and how it arises from concrete questions.
We then turn to interpreting this classical BV theory using dg Lie algebras and derived geometry
(i.e., we identify the moduli space it encodes).
Finally, we conclude by sketching how our theory appears as the chiral sector of a degeneration of the usual bosonic string when the target is a complex manifold with a Hermitian metric.
Our theory thus does provide insights into the usual bosonic string;
moreover, it clarifies why so many aspects of the bosonic string,
like the anomalies or $B$-fields, 
have holomorphic analogues.

\subsection{The theory we study} 

Let $V$ denote a complex vector space (the target),
and let $\langle-,-\rangle_V$ denote the evaluation pairing between $V$ and its linear dual~$V^\vee$.
Let $\Sigma$ denote a Riemann surface (the source).
Let $T_\Sigma^{1,0}$ denote the holomorphic tangent bundle on $\Sigma$, 
let $\langle-,-\rangle_T$ denote the evaluation pairing between $T_\Sigma^{1,0}$ and its vector bundle dual~$T_\Sigma^{1,0*}$. 
These are the key geometric inputs.

In a BV theory, the fields are $\ZZ$-graded;
we call this the {\em cohomological grading}.
We have four kinds of fields:
\[
\begin{array}{ccccc}
\text{field} & -1 & 0 & 1 & 2\\
\hline
\gamma & & \Omega^{0,0}(\Sigma) \otimes V & \Omega^{0,1}(\Sigma) \otimes V & \\
\beta & & \Omega^{1,0}(\Sigma) \otimes V^\vee & \Omega^{1,1}(\Sigma) \otimes V^\vee & \\
c & \Omega^{0,0}(\Sigma, T^{1,0}_\Sigma) & \Omega^{0,1}(\Sigma, T^{1,0}_\Sigma) & \\
b & & & \Omega^{1,0}(\Sigma, T^{1,0 *}_\Sigma) & \Omega^{1,1}(\Sigma, T^{1,0 *}_\Sigma)
\end{array}
\]
More accurately, we have eight different kinds of fields, 
but we view each row as constituting a single type 
since each given row consists of the Dolbeault forms of a holomorphic vector bundle.
For instance, the field $\gamma$ is a $(0,*)$-form with values in the trivial bundle with fiber~$V$,
and the field $b$ is a $(0,*)$-form with values in the bundle $T^{1,0 *} \otimes~T^{1,0 *}$.

To orient oneself it is helpful to start by examining the fields of cohomological degree zero,
since these typically have a manifest physical meaning.
For instance, the degree zero $\gamma$ field is a smooth $V$-valued function
and hence the natural field for the nonlinear $\sigma$-model into~$V$.
The degree zero $c$ field is a smooth $(0,1)$-form with values in vector field ``in the holomorphic direction,''
and hence encodes an infinitesimal change of complex structure of~$\Sigma$.
The degree $-1$ part of $c$ contains the gauge fields of the theories, vector fields. 
The equations of motion dictate that these vector fields be holomorphic, so we are seeing the infinitesimal version of the symmetry by biholomorphisms we mentioned above.
These constitute the obvious fields to introduce for a holomorphic version of the bosonic string.
The fields $\beta$ and $b$ are less obvious but appear as ``partners'' (or antifields)
whose role is clearest once we have the action functional and hence equations of motion.

The action functional is
\begin{equation}\label{bosaction}
S(\gamma,\beta,c,b) = 
\int_\Sigma \langle \beta, \dbar \gamma \rangle_V 
+ \int_\Sigma \langle b, \dbar c \rangle_T 
+ \int_\Sigma \langle \beta, [c,\gamma] \rangle_V 
+ \int_\Sigma \langle b, [c,c] \rangle_T.
\end{equation}
(We discuss below how to think about fields with nonzero cohomological degrees as inputs.)
The equations of motion are thus
\begin{alignat*}{2}
0 &= \dbar \gamma + [c,\gamma] & \quad\quad  0 &= \dbar \beta + [c,\beta] \\
0  &= \dbar c + \tfrac{1}{2} [c,c] & \quad\quad  0 &= \dbar b + [c,b] 
\end{alignat*}
Note that these equations are familiar in complex geometry.
For instance, the equation purely for $c$ encodes a deformation of complex structure on $\Sigma$; concretely, it modifies the $\dbar$ operator to $\dbar + c$.
The other equations then amount to solving for holomorphic sections (of the relevant bundle) withe respect to this deformed complex structure.
For instance, the equation in $\gamma$ picks out holomorphic maps from $\Sigma$,
with the $c$-deformed complex structure, to~$V$.

The field $b$ can be understood as an ``antifield" to the ghost field $c$; in other words, it is an {\em antighost}.
In that sense, $b$ does not have any intrinsic, physical meaning by itself.  

\begin{rmk}
\label{rmk:bcbg}
Just looking at this action functional, one might notice that if one drops the last two terms,
which are cubic in the fields, then one obtains a free theory
\begin{equation}
S_{free}(\gamma,\beta,c,b) = 
\int_\Sigma \langle \beta, \dbar \gamma \rangle_V 
+ \int_\Sigma \langle b, \dbar c \rangle_T,
\end{equation}
which is known as the {\em free $bc\beta\gamma$ system}.
Thus, one may view the holomorphic bosonic string as a deformation of this free theory
by ``turning on'' those interaction terms.
We will repeatedly try a construction first with this free theory before tackling the string itself,
as it often captures important information with minimal work.
For instance, we will examine the vertex algebra for the free theory before seeing how the interaction affects the operator products.
Similarly, one can identify the anomaly already at the level of the free theory.

This viewpoint of arriving at the bosonic string as a deformation of a free CFT is central to the analysis of the string in the physics literature \cite{GSW1} and Chapter 2 of \cite{polchinski}. 
See also the work in \cite{Scherk}. 
\end{rmk}

\begin{rmk}
\label{rmk:curved}
It is easy to modify this action functional to allow a curved target,
i.e., one can replace the complex vector space $V$ with an arbitrary complex manifold~$X$. 
The fields $b,c$ remain the same.
The degree 0 field $\gamma$ still encodes smooth maps into $X$, but now the degree 1 field is a section of $\Omega^{0,1}(\Sigma, \gamma^*T^{1,0}_X)$.
Similarly, $\beta$ is now a section of $\Omega^{1,*}(\Sigma, \gamma^*T^{1,0*}_X)$.
The action is then
\begin{equation}\label{curved action}
S(\gamma,\beta,c,b) = 
\int_\Sigma \langle \beta, \dbar \gamma \rangle_{T_X}
+ \int_\Sigma \langle b, \dbar c \rangle_{T_\Sigma} 
+ \int_\Sigma \langle \beta, [c,\gamma] \rangle_{T_X}
+ \int_\Sigma \langle b, [c,c] \rangle_{T_\Sigma}.
\end{equation}
In Section \ref{sec:curved} we will indicate how the results with linear target generalize to this situation.
\end{rmk}

\subsection{From the perspective of derived geometry}

We would like to explain what this theory is about in more conceptual terms,
rather than simply by formulas and equations.
Thankfully this theory is amenable to such a description.
We will be informal in this section and not specify a particular geometric context (e.g., derived analytic stacks),
except when we specialize to the deformation-theoretic situation (i.e., perturbative setting) that is our main arena.

\def\Maps{\operatorname{Maps}}

Let $\cM$ denote the moduli space of Riemann surfaces,
so that a surface $\Sigma$ determines a point in~$\cM$.
Let $\Maps_{\dbar}(\Sigma,V)$ denote the space of holomorphic maps from $\Sigma$ to $V$,
and hence a bundle $\Maps_{\dbar}(-,V)$ over $\cM$ by varying~$\Sigma$.
For our equations of motion, the $\gamma$ and $c$ fields of a solution determine a point in this bundle~$\Maps_{\dbar}(-,V)$. 
The commutative algebra $\cO(\Maps_{\dbar}(\Sigma,V))$ of functions on the space encodes the  observables of the classical theory.

\def\RS{{\cR\cS}}

This construction makes sense on noncompact Riemann surfaces as well.
Let $\RS$ denote the category whose objects are Riemann surfaces and whose morphisms are holomorphic embeddings.
There is a natural site structure: a cover is a collection of maps $\{S_i \to \Sigma\}_i$ such that the union of the images is all of~$\Sigma$.
Then $\Maps_{\dbar}(-,V)$ defines a sheaf of spaces over~$\RS$.
The observables for the classical theory is, in essence, the {\em co}\/sheaf of commutative algebras~$\cO(\Maps_{\dbar}(-,V))$,
and hence provides a factorization algebra.

In fact, it is better to work with the derived version of these spaces.
One important feature of derived geometry is that the appropriate version of a tangent space at a point is, in fact, a cochain complex.
In our setting, a point $(c,\gamma)$ in $\Maps_{\dbar}(-,V)$ determines a complex structure $\dbar + c$ on $\Sigma$---we denote this Riemann surface by $\Sigma_c$---and $\gamma$ a $V$-valued holomorphic function on~$\Sigma_c$.
The tangent complex of $\Maps_{\dbar}(-,V)$ at $(c,\gamma)$ is precisely 
\[
\Omega^{0,*}(\Sigma_c,T^{1,0})[1] \oplus \Omega^{0,*}(\Sigma_c,V).
\]
The first summand is the usual answer from the theory of the moduli of surfaces 
(recall, for example, that the ordinary tangent space is the sheaf cohomology $H^1(\Sigma,T_\Sigma)$ of the holomorphic tangent sheaf),
and the second is usual elliptic complex encoding holomorphic maps.

\begin{rmk}
It is useful to bear in mind that the degree zero cohomology of the tangent complex will recover the ``naive'' tangent space. 
In our case, we have 
\[
H^1(\Sigma_c,T_{\Sigma_c}) \oplus H^0(\Sigma_c,V),
\]
which encodes deformations of complex structure and holomorphic maps.
Negative degree cohomology of the tangent complex detects infinitesimal automorphisms (and automorphisms of automorphisms, etc) of the space.
For instance, here we see $H^0(\Sigma_c,T_{\Sigma_c})$ appear in degree -1, 
since a holomorphic vector field is an infinitesimal automorphism of a complex curve.
These negative directions are called ``ghosts'' (or ghosts for ghosts, and so on) in physics.
The positive degree cohomology detects infinitesimal relations (and relations of relations, and so on).
\end{rmk}

Note that the underlying graded spaces of this tangent complex are the $c$ and $\gamma$ fields from the BV theory described above.
We emphasize that the tangent complex is only specified up to quasi-isomorphism,
but it is compelling that a natural representative is the BV theory produced by the usual physical arguments.
This behavior, however, is typical of the relationship between derived geometry and BV theories:
when physicists write down a classical BV theory, 
the underlying free theory is essentially always the tangent complex of a nice derived stack.

The reader has probably noticed that, yet again, we have postponed discussing the $\beta$ and $b$ fields.
From a derived perspective, the full BV theory describes the shifted cotangent bundle $\TT^*[-1]\Maps_{\dbar}(-,V)$.
At the level of a tangent complex, the shifted cotangent direction contributes
\[
\Omega^{1,*}(\Sigma_c,T^{1,0*})[-1] \oplus \Omega^{1,*}(\Sigma_c,V^\vee),
\]
whose underlying graded spaces are the $\beta$ and $b$ fields.
These ``antifields'' are added so that the overall space of fields has a 1-shifted symplectic structure  when $\Sigma$ is closed, and a shifted Poisson structure when $\Sigma$ is open.

\subsection{Relationship to the Polyakov action functional}

This holomorphic bosonic string has a natural relationship with the usual bosonic string.
We sketch it briefly, only considering a linear target.

We begin with a bosonic string theory where the source is a 2-dimensional smooth oriented Riemannian manifold $\Sigma$ and the target is a Hermitian vector space~$(V,h)$. 
The ``naive'' action functional is
\ben
S^{\text{naive}}_{Poly}(\varphi, g) = \int_\Sigma h(\varphi, \Delta_{g} \varphi)\, \dvol_g
\een
where the field $g$ is a Riemannian metric on $\Sigma$ and the field $\varphi$ is a smooth map from $\Sigma$ to~$V$.
The notation $\Delta_g$ denotes the Laplace-Beltrami operator on~$\Sigma$. 

Note that $S^{naive}_{Poly}$ is invariant under the diffeomorphism group ${\rm Diff}(\Sigma)$ and under rescalings of the metric
(i.e., the theory is classically conformal).
Typically we express these Weyl rescalings as $g \mapsto e^{f} g$ with $f \in~C^\infty(\Sigma)$.
As we are interested in a string theory, we want to gauge these symmetries.
In geometric language, we want to think about the quotient stack 
obtained by taking solutions to the equations of motion and quotienting by these symmetry groups.

Our focus is perturbative, so that we want to study the behavior of this action near a fixed solution to the equations of motion.
In other words, we want to work with the Taylor expansion of the true action near some solution.
Hence, we work around a fixed metric $g_0$ on $\Sigma$, 
and we substitute for the field $g$,
the term $g_0+\alpha$ where $\alpha \in \Gamma(\Sigma,\Sym^2(T_\Sigma))$.
That is, we will consider deformations of~$g_0$.
As $\varphi$ is linear, we just consider expanding around the zero map.
Thus our new fields are $\varphi \in C^\infty(\Sigma,V)$ and $\alpha \in \Gamma(\Sigma,\Sym^2(T_\Sigma))$.

There are also ghost fields associated to the symmetries we gauge.
First, there are infinitesimal diffeomorphisms,  which are described by vector fields on~$\Sigma$.
We denote this ghost field by $X \in \Gamma(\Sigma,T_\Sigma)$.
It acts on the fields by the transformation 
\[
(\varphi,\alpha) \mapsto (\varphi + X \cdot \varphi, \alpha + L_X \alpha), 
\]
where $L_X$ denotes the Lie derivative on tensors.
Second, there are infinitesimal rescalings of the metric, 
such as $\alpha \mapsto \alpha + f \alpha$, 
with ghost field $f \in~C^\infty(\Sigma)$.
The rescaling does not affect $\varphi$.
The two symmetries are compatible: 
given $f$ and $X$, then $L_{X} (f \alpha) = X(f) \alpha + f L_X \alpha$ for any $\alpha \in \Sym^2(T_\Sigma)$.

To summarize, we have the following graded vector space of fields:
\[
\begin{array}{ccccc}
\text{field/antifield} & -1 &0 & 1 & 2\\
\hline
\varphi, \varphi^\vee & & \Omega^{0}(\Sigma) \tensor V & \Omega^2(\Sigma) \tensor V & \\
\alpha, \alpha^\vee & &  \Omega^0(\Sigma,\Sym^2(T_\Sigma)) & \Omega^2(\Sigma ; \Sym^{2}(T^*_\Sigma)) & \\
X, X^\vee & {\rm Vect}(\Sigma) & & & \Omega^2(\Sigma ; T^*_\Sigma)\\
f, f^\vee & C^\infty(\Sigma) &&& \Omega^2(\Sigma) .
\end{array}
\]
The BV action functional is of the form:
\begin{align}\label{polyakov bv}
S^{BV}_{Poly} (\varphi, \alpha, X, f) = 
& \int_\Sigma h(\varphi, \Delta_{g_0} \varphi)\, \dvol_{g_0} + \sum_{n \geq 1} \frac{1}{n!} \int_{\Sigma} h(\varphi, D_n(\alpha) \varphi) \dvol_{g_0} \\
& +  \int_{\Sigma} h(\varphi, X \cdot \varphi) \dvol_{g_0} \\
& +  S'(X, f, \alpha) 
\end{align}
The right hand side of the first line amounts to expanding out the Laplace-Beltrami operator $\Delta_{g_0 + \alpha}$ as a function of $\alpha$.
Hence, the $D_n$ are differential operators of the form 
$D_n : \left(\Sym^2(T_\Sigma)\right)^{\tensor n} \to {\rm Diff}^{\leq 2} (\Sigma)$ 
where ${\rm Diff}^{\leq 2}(\Sigma)$ are order $\leq 2$ differential operators on $\Sigma$.
Thus, for each section $\alpha$ of $\Sym^2(T_\Sigma)$, we get a second-order differential operator $D_n(\alpha)$ acting on functions on $\Sigma$. (This term is the $n$th term in the Taylor expansion, so its homogeneous of order $n$: $D_n(t \alpha) = t^n D_n(\alpha)$ for a scalar $t$.)
The second line encodes how vector fields act on the maps of the $\sigma$-model.
The third line $S'(X, f, \alpha)$ is independent of $\varphi$
and only depends on the fields $f$, $X$, $\alpha$ and their antifields (denoted with checks $(-)^\vee$).
It is of the form
\ben
S'(f,X,\alpha) = \int_\Sigma \<\alpha^\vee, L_X (g_0 + \alpha) + f (g_0 + \alpha)\> + \int_\Sigma \<X^\vee, [X,X]\> + \int_\Sigma \<f^\vee, X \cdot f\> .
\een
The first term encodes how vector fields and Weyl transformations act on the perturbed metric $g_0 + \alpha$ and the remaining terms are required to ensure the gauge symmetry is consistent (satisfies the classical master equation). 

An explicit formula for $D_n(\alpha,\ldots,\alpha)$ is a rather involved exercise (and not needed here).
For instance, if we are working locally on $\Sigma = \RR^2$ with the $g_0$ the flat metric, 
then the operator $D_1(\alpha)$ is sum of a first-order and a second-order differential operator
\ben
D_1(\alpha) = \frac{1}{2} \frac{\partial}{\partial x^i} ({\rm tr}(\alpha)) \frac{\partial}{\partial x^i} + \frac{1}{2} {\rm tr} (\alpha) \frac{\partial}{\partial x^i} \frac{\partial}{\partial x^i}, 
\een
or in a more coordinate-free notation, 
\[
D_1(\alpha) = \frac{1}{2} \star \d \left({\rm tr}(\alpha) \star \d\right).
\] 
Here, we use the natural trace map ${\rm tr} : \Sym^2 T\Sigma \to C^\infty(\Sigma)$ of symmetric $2\times2$ matrices. 

There is an important parameter in this action functional: the Hermitian inner product $h$.
We can consider scaling it $h \to t h$, with $t \in (0,\infty)$.
The ``infinite volume limit'' as $t \to \infty$ admits a nice description,
provided one rewrites the action functional in a first-order formalism
(i.e., adjoins fields so that only first-order differential operators appear in the action,
which is a sort of action functional analogue of working with phase space).

\begin{lem} 
In this infinite volume limit
the bosonic string becomes equivalent to a BV theory whose action functional has the form
\ben
S(\beta, \gamma, b,c) + \Bar{S}(\Bar{\beta}, \Bar{\gamma}, \Bar{b}, \Bar{c}),
\een
where $S(\beta, \gamma, b,c)$ is the action functional for the holomorphic bosonic string in Equation (\ref{bosaction}) and $\Bar{S}$ is its anti-holomorphic conjugate. 
\end{lem}

\begin{rmk} 
The action functional $\Bar{S}$ is similar to $S$ where the fields $\gamma,\beta,b,c$ are replaced by sections in the the relevant conjugate bundles. 
For example, $\beta \in \Omega^{1,*}(\Sigma)$ becomes $\Bar{\beta} \in \Omega^{*,1} (\Sigma)$. 
Moreover, the operator $\dbar$ is replaced by the holomorphic Dolbeault operator $\partial$. 
Another way of saying this is that $\Bar{S}$ is the holomorphic string on $\Bar{\Sigma}$, which is the conjugate complex structure to $\Sigma$. 
\end{rmk}

\begin{rmk}
For physics references that study the holomorphic splitting of the Polyakov action (and more general conformal field theories), 
we refer to the original sources~\cite{Belavin, Laz}.
\end{rmk}

\begin{proof}[Outline of proof] 
There are two things that may cause alarm in the statement of the claim. 
First, the space of fields of the Polyakov string (in the BV language) and those of the holomorphic bosonic string do not match up. 
Second, the infinite volume limit $t \to \infty$ is naively ill-defined using the action functional (\ref{polyakov bv}). 
It turns out that these two issues are solved by the same maneuver. 

We begin with the first term in the first line of (\ref{polyakov bv}). 
Notice that it is simply the action functional for the $\sigma$-model of maps from $(\Sigma, g_0)$ to $(V, h)$. 
It is shown in Appendix 21 of \cite{GGW} how to make sense of the infinite volume limit of this usual $\sigma$-model. 
The idea is to rewrite this theory in the {\em first order formalism}.
This amounts to introducing a new field $B \in \Omega^1(\Sigma) \tensor V^\vee$ and action functional 
\ben
\int_\Sigma \<B, \d \varphi\>_V - \frac{1}{2} \int_\Sigma h^\vee(B, \star B)
\een
where $\<-,-\>_V$ represents the evaluation pairing between $V$ and its dual, 
$\star$ is the Hodge star operator for the metric $g_0$, 
and $h^\vee$ denotes the dual metric on $V$. 
This action functional is equivalent to the original $\sigma$-model;
one can compare the equations of motion. 
Moreover, since $(th)^\vee = (1/t)h^\vee$, 
the dual $(th)^\vee$ goes to $0$ in the infinite volume limit $t \to \infty$, 
and hence kills the second term in the first order action. 
The remaining theory splits as the direct sum of the free $\beta\gamma$ system with target $V$ and its anti-holomorphic conjugate. 
At the level of fields, the original field $\varphi$ corresponds to $\gamma + \Bar{\gamma}$ in the first order description,
and $B$ corresponds to~$\beta+\Bar{\beta}$. 

We now consider the remaining terms in the first line of the action (\ref{polyakov bv}). 
Note that this action only depends on the conformal class of the metric, 
i.e., on the metric up to Weyl rescaling.
Hence this feature remains true in the infinite volume limit,
which simplifies the situation considerably, as we now explain.

It is convenient to work in holomorphic coordinates for the complex structure determined by the background metric $g_0$. 
With respect to this complex structure, the tensor square of the cotangent bundle splits as
\ben
T^* \tensor T^* = T^{*0,1}\tensor T^{1,0} \oplus T^{*1,0} \tensor T^{0,1} \oplus T^{*1,0} \tensor T^{*0,1} \oplus T^{*0,1} \tensor T^{*1,0} .
\een
Sitting inside of this bundle is the symmetric square, where the field $\alpha$ lives. 
With respect to this splitting, we write sections as $\alpha = c + \Bar{c} + f g_0$, where $f \in C^\infty(\Sigma)$. 
But since the action only depends on the conformal class of the metric, 
only the fields $c$ and $\Bar{c}$ are relevant.
In the first-order formalism of the preceding paragraph, 
we thus find that the remaining terms in the first line of (\ref{polyakov bv}) read 
\ben
\int_\Sigma \<\beta , [c, \gamma]\>_V + \int_\Sigma \<\Bar{\beta} , [\Bar{c}, \bar{\gamma}]\>_V.
\een
This first term is precisely the third term in the holomorphic string action functional (\ref{bosaction}), 
which describes how deformations of complex structure couple to the fields of the $\sigma$-model. 

In the infinite volume limit, the term $S'(f, X, \alpha)$ recovers the terms
\[
\int_\Sigma \langle b, \dbar c \rangle_T + \int_\Sigma \langle b, [c,c] \rangle_T
\]
in the action of the holomorphic string, plus their conjugates. 
The arguments are similar to those we have just sketched. 
\end{proof}

\begin{rmk} Another approach to arrive at the holomorphic theory we consider comes from considering supersymmetry. 
Without gravity, the pure holomorphic $\sigma$-model can be viewed as the {\em holomorphic twist} of the $N=(2,0)$ supersymmetric $\sigma$ model (in this case the target is required to be K\"{a}hler). 
Moreover, the $\beta\gamma bc$ system is the holomorphic twist of the $N=(2,2)$ model. 
Conjecturally, we expect the holomorphic theory of gravity we consider to be the holomorphic twist of two-dimensional $N=2$ supergravity.
\end{rmk}

\begin{rmk} In this infinite volume limit, one can put the dependence of the target metric back into the theory by choosing a certain background to work in. 
In the BV formalism this amounts to choosing a certain deformation parameter, which in this instance corresponds to infinitesimal deformations of the target metric.
Note that to deform the metric on the target we leave the world of ``holomorphic field theory" as the deformation involves both $z$ and $\zbar$ dependent terms. 
It would be interesting to study how to formulate the theory with finite target metric in the BV formalism.
\end{rmk}

\section{Deformations of the theory and string backgrounds}
\label{sec: moduli}

Whenever one is studying a theory,
it is helpful to understand how it can be modified 
and how features of the theory change as one adjusts natural parameters of the theory,
such as coupling constants of the action functional.
In other words, one wants to understand the theory in the moduli space of classical theories.

In the BV formalism, because we are working homologically, this moduli space is derived,
and there is a tangent complex to our theory in the moduli of classical BV theories.
We call it the {\em deformation complex} of the theory.
A systematic discussion can be found in Chapter 5 of~\cite{CosBook}.

As a gloss, the underlying graded vector space of this deformation complex consists of the local functionals on the jets of fields, i.e., Lagrangian densities.
(Note that we allow local functionals of arbitrary cohomological degree.) 
There is also a shifted Lie bracket $\{-,-\}$, 
which arises from the pairing $\int_\Sigma \langle-,-\rangle$ on the fields.
It is, in essence, the shifted Poisson bracket corresponding to that shifted symplectic pairing on the fields.
The differential on the local functionals is then $\{S,-\}$, where $S$ is the classical action. 
All together, the deformation complex forms a shifted dg Lie algebra. 
Observe that if we find a degree zero element $I$ such that
\[
0=\{S +I,S +I\}=2\{S,I\}+\{I,I\},
\]
then $I$ is a shifted Maurer-Cartan element and 
hence determines a new classical BV theory whose action functional is $S + I$. 
In particular, degree 0 cocycles determine first-order deformations of the classical BV theory. Cocycles in degree -1 encode local symmetries of the classical theory; 
and obstructions to satisfying the quantum master equation end up being degree 1 cocycles.

In this section, we will explain why the deformation complex $\Def_{\rm string}$ of the holomorphic string 
can be expressed in terms of Gelfand-Fuks cohomology \cite{Fuks}. 
Along the way we will see how the usual backgrounds for the bosonic string (a target metric, dilaton term, and so on) appear as elements in this complex of local functionals and hence as deformations of the classical action. 

Right now, we will focus on the case $\Sigma = \CC$, 
and in Section \ref{sec: conformalblock} we will consider arbitrary Riemann surfaces.
We restrict ourselves to examining {\em translation-invariant} local functionals (which will allow us to descend to a theory defined on an elliptic curve).
Unpacking what this means will lead swiftly to Gelfand-Fuks cohomology.

\subsection{Deformations for the classical theory}

As a local functional is given by integration of a Lagrangian density,
translation invariance requires the density to be the Lebesgue measure $\d^2 z$, up to rescaling,  
and requires the Lagrangian to be specified by its behavior at one point.
Hence, a translation-invariant local functional on $\CC$ is determined by a function of the jet (i.e., Taylor expansion) of the fields at the origin in~$\CC$. 

It is particularly easy to understand what we mean in the case of the free $bc\beta\gamma$ system.
For instance, the $\gamma$ fields live in the Dolbeault complex $\Omega^{0,*}(\CC ; V)$,
and their jets at the origin are $(V [[z,\zbar]] [\d \zbar] , \dbar)$,
where $\dbar$ is the formal Dolbeault differential. 
An example of an element is thus $\widehat{\gamma} = \sum_{m,n} \frac{1}{m!n!}g_{mn} z^m \zbar^n$,
which is just a formal power series with values in $V$.
An example of a functional is
\[
F(\widehat{\gamma}) = g_{10} + g_{21} = \left( \partial_z \widehat{\gamma}\right)|_0 + \left( \partial_z^2 \partial_{\zbar} \widehat{\gamma}\right)|_0,
\]
which corresponds to the local functional
\[
F(\gamma) = \int_\CC \partial_z \gamma + \partial_z^2 \partial_{\zbar} {\gamma}\,\d^2 z.
\]
We call the first kind of term a {\em chiral} interaction, as it only depends on holomorphic derivatives.

By the $\dbar$-Poincar\'{e} lemma, 
this complex $(V [[z,\zbar]] [\d \zbar] , \dbar)$ is quasi-isomorphic to $V[[z]]$, concentrated in degree zero. 
This observation is actually quite concrete:
it simply says that for a solution $\gamma$ to the equation of motion $\dbar \gamma = 0$, 
its Taylor expansion is just a power series in $z$ and it is independent of $\zbar$.
In consequence, if we consider translation-invariant Lagrangians depending only on the $\gamma$ field, then up to quasi-isomorphism these are $\Sym(V^\vee[z^\vee])$.
In other words, only chiral interactions yield distinct modifications of the action,
when one takes into account the equation of motion.

Note that we have chosen to work with functionals of the fields
that are polynomials built out of continuous linear functionals $V^\vee[z^\vee]$ of the jets.
This choice is the standard and natural one for variational problems.
We note as well that constant functionals are irrelevant,
so we want to use $\Sym^{>0}(V^\vee[z^\vee])$ to describe translation-invariant local functionals.

An analogous argument applies to the $c$ field. 
It shows there is a quasi-isomorphism of dg Lie algebras 
between the jet at the origin of the Dolbeault complex $\Omega^{0,*}(\CC ; T^{1,0}_\CC)$ of holomorphic vector fields 
and the Lie algebra of formal vector fields $\wone = \CC[[z]]\partial_z$.
The translation-invariant Lagrangians depending only on the $c$ field 
are thus quasi-isomorphic to $\cred^*(\wone)$,
by which we mean the (reduced) {\em continuous} Lie algebra cohomology,
often known as the Gelfand-Fuks cohomology 
Similar arguments work for the $\beta$ and $b$ fields.

If we take all the fields into account together and consider the full equations of motion 
for the holomorphic string,
which couple the $c$ field to the others,
then these arguments yield the following.

\begin{lem}\label{lem: gf}
There is a quasi-isomorphism 
\[
\Def_{\rm string}(\CC,V)^\CC \simeq \cred^*(\wone, \Sym(V^\vee[z^\vee] \oplus V[z^\vee] \d z^\vee \oplus W_1^{\rm ad}[2])) [2]
\]
between the deformation complex of translation-invariant local functionals for the holomorphic string and a certain Gelfand-Fuks cochain complex.
\end{lem}

This lemma already substantially simplifies our lives, 
as one can invoke the literature on Gelfand-Fuks cohomology.
But before we do,
we will take advantage of another symmetry condition to simplify the situation.

\subsection{Dilating cotangent fibers}

We have already seen how to think of the holomorphic bosonic string theory 
as corresponding to the shifted cotangent bundle $\TT^*[-1]{\rm Maps}_{\dbar}(-, V)$, 
as a bundle over the moduli of Riemann surfaces. 
There is a natural action of the group $ \CC^\times$ on this space
by scaling the shifted cotangent fibers,
and we will use the notation $\CC^\times_{\rm cot}$ to indicate this appearance of the multiplicative group.

This group action can be seen on the level of the field theory as follows: 
we give the $\gamma$ and $c$ fields---the base of the cotangent bundle---weight $0$ and give the $\beta$ and $b$ fields---the cotangent fiber---weight~$1$. 
Note that, in consequence, the pairing $\langle-,-\rangle$ on fields thus has weight -1.
In these terms, the classical action functional is weight 1. 
Thus, we focus on weight 1 deformations of the action for the holomorphic bosonic string,
as we are interested in local functionals of the same kind.
That means we consider the subcomplex of weight 1 local functionals inside the deformation complex.

\begin{rmk}\label{rmk: classical weights}
Although this action $S$ has weight 1, 
its role in the cochain complex of classical observables is to define the differential $\{S,-\}$.
Observe that the shifted Poisson bracket $\{-,-\}$ has weight -1, 
because it is determined by the pairing, 
and so the differential has weight 0. 
\end{rmk}

This subcomplex admits a nice description in terms of the geometry of the target.

\def\wt{{\rm wt}}

\begin{lem}
\label{lem: def complex wt zero} 
There is a $\GL(V)$-equivariant quasi-isomorphism
\[
\Def_{\rm string}(\CC)^{\CC, \wt(1)} \simeq \Sym(V^*) \tensor V[1]
\]
between the weight 1, translation-invariant deformation complex 
and the polynomial vector fields on $V$, placed in degree~-1.
\end{lem}

Concretely, this result says that there are no weight zero interactions that are not not trivialized by an automorphism of the theory.
This claim is a consequence of the fact that the zeroth cohomology group vanishes.
On the other hand, this lemma says the theory admits a large group of symmetries,
namely diffeomorphisms of the target, 
which appears as the degree -1 cohomology.

The $\GL(V)$ equivariance takes into account the natural symmetries of the target. 
It also is the first step in the approach to studying the deformation complex with general curved target. 
We will discuss this further in the section on string backgrounds. 

\subsection{Interaction terms that appear at one loop}

As we will see in Section \ref{sec: quantization}, 
the quantization of the holomorphic string only involves local functional of weight zero for this $\CC^\times_{\rm cot}$-action.
(Concretely, this restriction appears because the one-loop Feynman diagrams only have external legs for $c$ and $\gamma$ fields.)
Hence, it behooves us to compute the weight zero subcomplex of the deformation complex as well.

\begin{lem}
\label{lem: def complex wt zero} 
There is a $\GL(V)$-equivariant quasi-isomorphism
\[
\Def_{\rm string}(\CC)^{\CC, \wt(0)} \simeq \CC[-1] \oplus \Omega^2_{cl}(V)[1] \oplus \Omega^1(V) \oplus \Omega^1_{cl}(V)[-1] 
\]
between the weight 0, translation-invariant deformation complex 
and natural complexes related to the geometry of the target.
\end{lem}

Before explaining the key steps of the proof, 
we remark that there is another, more structural way to see that only weight zero local functionals should be relevant.
A quick physical argument would say that we want the path integral measure $\exp(-S/\hbar)$ to be weight zero,
which forces $\hbar$ to have weight one to cancel out with the weight of the classical action.
But the one-loop term $I_1$ in the quantized action $S^\q = S + \hbar I_1 + \cdots$ must then have weight zero.

There is a BV analogue of this argument.
It notes that the differential of the quantum observables has the form $\{S^\q,-\} + \hbar \Delta$,
where $\Delta$ denotes the BV Laplacian.
(See Section \ref{subsec: QME} for a discussion of these objects.)
As the BV Laplacian has weight -1 because it is determined by the bracket,
we must give $\hbar$ weight 1 to ensure the total differential has weight zero.
Again the one-loop interaction is forced to have weight zero.

\subsubsection{Sketch of proof}
We have already mentioned that we can identify the full translation-invariant deformation complex with a certain Gelfand-Fuks cohomology. 
In terms of this Gelfand-Fuks cohomology we find that the cotangent weight zero piece is identified with 
\ben
\Def_{\rm string}(\CC)^{\CC, \wt(0)} = \cred^*\left(\wone ; \Sym(V^\vee[z^\vee]) \right) [2]. 
\een 
We will drop the overall shift by $2$ until the end of the proof. 

Any symmetric algebra has a natural maximal ideal: 
for any vector space $W$,
\[
\Sym(W) = \CC \oplus \Sym^{\geq 1}(W).
\] 
Thus, we can decompose our complexes as
\ben
 \cred^*\left(\wone ; \Sym(V^\vee[z^\vee]) \right) =   \cred^*(\wone) \oplus \clie^*\left(\wone ; \Sym^{\geq 1} (V^\vee[z^\vee]) \right) .
 \een
The first summand is the reduced Gelfand-Fuks cohomology of formal vector fields with values in the trivial module.
It is well-known that $H^3_{\rm red} (\wone) \cong \CC[-3]$, 
i.e., this cohomology is one-dimensional and concentrated in degree~$3$. 

We now proceed to computing the second summand. 
Denote by $\{L_n = z^{n+1} \partial_z\}$ the standard basis for the Lie algebra of formal vector fields $\wone$. 
Notice that the Euler vector field $L_0 = z \partial_z$ induces a grading on $\wone$,
that we will call {\em conformal dimension}.
Note that $L_n$ has conformal dimension $n$. 
This grading extends naturally to the Chevalley-Eilenberg complex of $\wone$ with coefficients in any module. 

Let $\lambda_n \in \wone^\vee$ be the dual vector to $L_n$. 
(We work with the continuous dual vector space, as in the setting of Gelfand-Fuks cohomology.) 
An arbitrary element of $V [[z]]$ is linear combination of vectors of the form $v \tensor z^k$. 
Write $\zeta_k$ for the dual element $(z^k)^\vee$. 
Thus an element of $(V [[z]])^\vee$ is a linear combination of the vectors of the form $v^\vee \tensor \zeta_k$. 

 \begin{lem} \label{lem: gf}
 Let $M$ be any $\wone$-module. 
The inclusion of the subcomplex of conformal dimension zero elements
\ben
\clie^*(\wone ; M)^{\wt(0)} \xto{\simeq} \clie^*(\wone ; M)
\een
is a quasi-isomorphism.
\end{lem}

\begin{proof} 
For each $p-1 \geq 0$, 
define the operator $\iota_{L_0} : \clie^{p}(\wone ; M) \to \clie^{p-1}(\wone ; M)$ by sending a cochain $\varphi$ to the cochain
\ben
(\iota_{L_0}\varphi)(X_1,\ldots,X_p) = \varphi(L_0, X_1,\ldots,X_p) .
\een 
Let $\d$ be the differential for the complex $\clie^*(\wone ; M)$. It is easy to check that the difference $\d \iota_{L_0} - \iota_{L_0} \d$ is equal to the projection onto the dimension zero subspace. 
\end{proof}

The underlying graded vector space of this conformal dimension zero subcomplex splits as follows:
\be\label{splitting}
\clie^{\#}(\wone)^{\wt(0)} \tensor \left(\Sym^{\geq 1}\left(V [[z]]\right)^\vee \right)^{\wt(0)} \oplus \clie^{\#}(\wone)^{\wt(1)} \tensor \left(\Sym^{\geq 1}\left(V [[z]]\right)^\vee\right)^{\wt(-1)}
\ee
In the first component, the purely dimension zero part of the reduced symmetric algebra is simply $\Sym^{\geq 1}(V^\vee)$, i.e., power series 
on $V$ with no constant term.
We denote this algebra concisely as $\cO_{red}(V)$, for reduced functions on $V$.
Similarly, in the second component, 
the dimension one part of $\Sym^{\geq 1}\left(V[[z]]\right)^\vee$ is of the form ${\rm Sym}(V^\vee) \tensor z^\vee V^\vee$, which is naturally identified with~$\Omega^1(V)$. 

The differential in this Gelfand-Fuks complex has the form
\[
\xymatrix{
\overset{\ul{0}}{1 \otimes \sO_{red}(V)} \ar[rd]^{\d_{dR}} & \overset{\ul{1}}{\lambda^0 \otimes \sO_{red}(V)} \ar[r] \ar[rd]^{\d_{dR}} & \overset{\ul{2}}{\lambda^{-1} \wedge \lambda^1 \otimes \sO_{red}(V)} & \overset{\ul{3}}{\lambda^{-1} \wedge \lambda^1 \wedge \lambda^0 \otimes \sO_{red}(V)} \\
 & \lambda^{-1} \otimes \Omega^1(V) \ar[r] & \lambda^{-1} \wedge \lambda^0 \otimes\Omega^1(V) &
}
\]
The top line comes from the first summand in (\ref{splitting}) and the bottom line corresponds to the second summand.
The top horizontal map sends $\lambda^0$ to $2 \cdot \lambda^{-1} \wedge \lambda^1$, 
and the bottom horizontal map sends $\lambda^{-1}$ to $\lambda^{-1} \wedge \lambda^0$ (both are the identity on $V$). 
The diagonal maps are given by the de Rham differential $\d_{dR} : \sO_{red}(V) \to \Omega^1(V)$. 
This complex is quasi-isomorphic to 
\[
\xymatrix{
1 \otimes \sO_{red}(V) \ar[rd]^{\d_{dR}} & & & \lambda^{-1} \wedge \lambda^1 \wedge \lambda^0 \otimes \sO_{red}(V) \\
 & \lambda^{-1} \otimes \Omega^1(V) & \lambda^{-1} \wedge \lambda^0 \otimes\Omega^1(V) &
}
\]
which, in turn, is identified with $\Omega^{2}_{cl}(V)[-1] \oplus \Omega^1(V)[-2] \oplus \Omega^1_{cl}(V)[-3]$. 
After accounting for the overall shift by $2$, 
we arrive at the identification of the $\CC_{\rm cot}^\times$-weight zero component of the translation-invariant deformation complex.

\subsection{Interpretation as string backgrounds}

We now discuss, in light of the calculations above, how to interpret string backgrounds in our approach. 
Since $V$ is flat,
we will see that the following deformations will be trivializable. 
Note that this trivializations will {\em not} be equivariant for the obvious $\GL(V)$ action (or for non-flat targets, general diffeomorphisms of the target). 
Thus, these deformations are relevant for the case of a curved target, and we can give an interpretation of them in terms of the usual perspective of {\em string backgrounds}. 

We have already mentioned that we should think of the $\CC^\times_{\rm cot}$ weight $1$ local functionals as deformations of the classical theory as a cotangent theory.
The cohomological degree zero deformations of the weight one deformations is $H^1(V ; T_V)$. 
Given any such element $\mu \in H^1(V ; T_V)$ we can consider the following local functional
\ben
\int_\Sigma \<\beta, \mu(\gamma)\>_V .
\een 
The element $\mu$ determines a deformation of the complex structure of $V$, and we have prescribed an action functional encoding this deformation. 
We propose that this an appearance of the ordinary curved background in bosonic string theory from the perspective of the holomorphic model we work with.

There are interesting deformations that go outside of the world of cotangent theories. 
Consider the cohomological degree zero part of the weight 0 complex. 
There is a term of the form $H^1(V ; \Omega^2_{cl}(V))$.
It is shown in Part 2 Section 8.5 of \cite{GGW} how closed holomorphic two-forms determine local functionals of the $\beta\gamma$ system with curved target. 
A sketch of this construction goes as follows.
Locally we can write a closed holomorphic 2-form as $\d \theta$ for some holomorphic one-form $\theta \in \Omega^1(V)$. 
If $\gamma : \Sigma \to V$ is a map of the $\sigma$-model there is an induced map (when $\gamma$ satisfies the equations of motion) $\gamma^* : \Omega^1(V) \to \Omega^1(\Sigma)$. 
We can then integrate $\gamma^* \theta$ along any closed cycle $C$ in $\Sigma$ and one should think of this as a residue along $C$. 
In \cite{GGW} we write down a local functional that realizes this residue, and one can show that it only depends on the corresponding class in $H^1(V ; \Omega^2_{cl}(V))$. 
We posit that this is the appearance of the $B$-field deformation of the ordinary bosonic string. 

In future work we aim to study how our description of holomorphic string backgrounds compares to the approaches of string backgrounds in the physics literature. 
See for instance \cite{CFMP} for an overview.

\section{Quantizing the holomorphic bosonic string on a disk} 
\label{sec: quantization}

For us, quantization will mean that we use perturbative constructions in the setting of the BV formalism.
Concretely, this means that we enforces the gauge symmetries using the homological algebra of the BV formalism and that we use Feynman diagrams and renormalization to obtain an approximation for the desired, putative path integral. 
There are toy models for this approach where one can see very clearly how it gives asymptotic expansions for finite-dimensional integrals \cite{GJF}.
In particular, these toy models show that this approach need not recover the true integral
but does know important information about it;
a similar relationship should hold between this quantization method and the putative path integral, 
but in this case there is no {\em a priori} definition of the true integral in most cases.

This notion of quantization applies to any field theory arising from an action functional,
and the algorithm one applies to obtain a quantization is the following:
\begin{enumerate}
\item Write down the integrals labeled by Feynman diagrams arising from action functional.
\item Identify the divergences that appear in these integrals and add ``counterterms'' to the original action that are designed to cancel divergences.
\item Repeat these steps until no more divergences appear in Feynman diagrams.
We call this the ``renormalized action.''
\item Check if the renormalized action satisfies the quantum master equation. 
If it does, you have a well-posed BV quantum theory, and we call the result a {\em quantized action}. If not, guess a way to adjust the renormalized action and begin the whole process again.
\end{enumerate}
It should be clear that along the way, one makes many choices;
hence if a quantization exists, it may not be unique.
It is also possible that a BV quantization may not exist.

In this section we will apply the algorithm in the case of $\Sigma~=~\CC$.
For this theory we are lucky, however:
at one-loop the integrals that appear in our quantization from the Feynman diagrams do not have divergences,
so that renormalized action is easy to compute.
This aspect is the subject of the first part of this section.
(In Section \ref{sec: conformalblock} we will provide an argument based on deformation theory as to why quantizations exist on arbitrary Riemann surfaces.)
Moreover, it is easy to check whether the quantum master equation is satisfied,
and the answer is simple.
This aspect is the subject of the second part.
The results can be summarized as follows.

\begin{prop}
The holomorphic bosonic string with source $\CC$ and target $\CC^d$ admits a BV quantization
if $d = 13$.
This quantized action only has terms of order $\hbar^0$ and $\hbar$ (i.e., it quantizes at one loop).
\end{prop}

\subsection{The Feynman diagrams}

Let us describe the combinatorics of the Feynman diagrams that appear here
before we describe the associated integrals.

\subsubsection{}

The procedure constructs graphs out of a prescribed type of vertices and edges;
we must consider all graphs with such local structure.
The classical action functional determines the allowed kinds of vertices and edges.
The quadratic terms of the action tell us the edges;
each quadratic term yields an edge whose boundary is labeled by the two fields appearing in the term.
For us there are thus two types of edges: 
an edge that flows from $\beta$ to $\gamma$, 
and an edge that flows from $b$ to~$c$ displayed in Figure \ref{fig:props}.
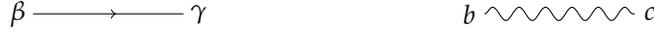
\begin{figure}
\begin{tikzpicture}
		\draw[fermion] (-4,0) -- (-2,0);
		\draw[vector] (2,0) -- (4,0);
		\draw (-4.2, 0) node {$\beta$};
		\draw (-1.8,0) node {$\gamma$};
		\draw (1.8,0) node {$b$};
		\draw (4.2,0) node {$c$};
\end{tikzpicture}
\caption{The $\beta\gamma$ and $bc$ propagators}
\label{fig:props}
\end{figure}

The nonquadratic terms tell us the vertices:
each $n$-ary term yields a vertex with $n$ legs,
and the legs are labeled by the $n$ types of fields appearing in the term.
For us there are thus two types of trivalent vertices:
a vertex with two $c$ legs and a $b$ leg, 
and a vertex with a $c$ leg, a $\gamma$ leg, and a $\beta$ leg.
It helpful to picture these legs as directed,
so that $c$ and $\gamma$ legs flow into a vertex
and $b$ and $\beta$ legs flow out. 
These vertices are displayed in Figure \ref{fig:verts}.

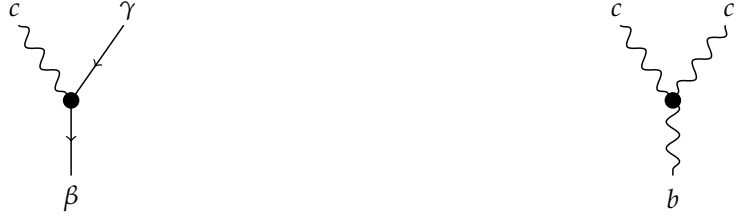
\begin{figure}
\begin{tikzpicture}[line width=.2mm, scale = 1]
		\draw[vector](-4.7,1)--(-4,0);
		\draw[fermion](-3.3,1)--(-4,0);
		\draw[fermion](-4,0)--(-4,-1);
		\filldraw[color=black]  (-4,0) circle (.1);
		\draw (-4.75,1.2) node {$c$};
		\draw (-3.25,1.2) node {$\gamma$};
		\draw (-4,-1.3) node {$\beta$};	
		
		\draw[vector](3.3,1)--(4,0);
		\draw[vector](4.7,1)--(4,0);
		\draw[vector](4,0)--(4,-1);
		\filldraw[color=black]  (4,0) circle (.1);
		\draw (3.25, 1.2) node {$c$};
		\draw (4.75,1.2) node {$c$};
		\draw (4, -1.3) node {$b$};
\end{tikzpicture}
\caption{The trivalent vertices for $\int \langle \beta, [c,\gamma] \rangle$ and $\int \langle b, [c,c] \rangle$}
\label{fig:verts}
\end{figure}

The kinds of graphs one can build with such vertices and edges are limited.
We focus on connected graphs, since an arbitrary graph is just a union of connected components.

A tree (i.e., a connected graph with no loops) must have at most one outgoing leg,
which must be either a $b$ or a~$\beta$;
the other legs are incoming, so each must be labeled by a $c$ or a~$\gamma$. 
An example of such a tree is given in Figure \ref{fig:tree}.

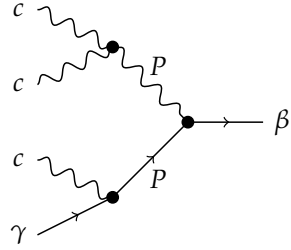
\begin{figure}
\begin{tikzpicture}[line width=.2mm, scale = 1]
		\draw[vector] (-2,3) -- (-1,2.5);
		\draw[vector] (-2, 2) -- (-1, 2.5);
		\draw[vector] (-1,2.5) --(0, 1.5);
		\draw (-2.25, 3) node {$c$};
		\draw (-2.25, 2) node {$c$};
		\draw(-.4, 2.25) node {$P$};
		\filldraw[color=black] (-1,2.5) circle (.075);
		\draw[vector] (-2, 1) -- (-1, .5);
		\draw[fermion] (-2, 0) -- (-1,.5);
		\draw[fermion] (-1,.5) -- (0,1.5);
		\draw (-2.25, 1) node {$c$};
		\draw (-2.25, 0) node {$\gamma$};
		\draw (-.4, .75) node {$P$};
		\filldraw[color=black] (-1,.5) circle (.075);
		\draw[fermion] (0,1.5) -- (1,1.5);
		\draw (1.25, 1.5) node {$\beta$};
		\filldraw[color=black] (0,1.5) circle (.075);
\end{tikzpicture}
\caption{An example of a tree with four inputs and one output}
\label{fig:tree}
\end{figure}
		
Note that there are two types of trees.
If there is a $\gamma$ leg, then there is a $\beta$ leg,
and there is a chain of $\gamma\beta$ edges connecting them;
all other external legs are of $c$~type.
If there is a $b$ leg, then the only other legs are $c$~type.

A one-loop graph will consist of a wheel (i.e., a sequence of edges that form an overall loop)
with trees attached.
The outer legs are all of $c$~type.
Every edge along a wheel will have the same type.
It is not possible to build a connected graph with more than one loop.
This combinatorics is the essential reason that we can quantize at one loop.
For an example of such a wheel see Figure \ref{fig:wheel}.

\begin{figure}
\begin{tikzpicture}[line width=.2mm, scale=1.5]
		\draw[vector](145:1) -- (145:.3cm);
			\node at (145:1.15) {$c$};
		\draw[vector](215:1) -- (215:.3cm);
			\node at (215:1.2) {$c$};
		\draw[vector](35:1) -- (35:.3cm);
			\node at (35:1.15) {$c$};
		\draw[vector](-35:1) -- (-35:.3cm);
			\node at (-35:1.2) {$c$};
		\draw[fill=black] (0,0) circle (.3cm);
		\draw[fill=white] (0,0) circle (.29cm);
	    	\clip (0,0) circle (.3cm);
\end{tikzpicture}
\caption{An example of a wheel with four inputs}
\label{fig:wheel}
\end{figure}
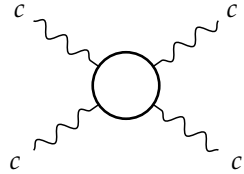

We write ${\bf Graph}_{\rm string}$ for the collection of connected graphs just described,
namely the directed trees and 1-loop graphs allowed by the string action functional.
Let ${\bf Graph}_{\rm string}^{(0)}$ denote the 0-loop graphs (i.e., trees) and let ${\bf Graph}_{\rm string}^{(1)}$ denote the 1-loop graphs (i.e., wheels with trees attached).

\subsubsection{}

These graphs describe linear maps associated to the field.
More precisely, a graph with $k$ legs describes a linear functional on the $k$-fold tensor product of the space of fields.
One builds this linear functional out of the data of the action functional.

As an example, a $k$-valent vertex corresponds to a $k$-ary term in the action,
which manifestly takes in $k$ copies of the fields and outputs a number.
Thus, the vertex labels an element of a (continuous) linear dual of the $k$-fold tensor product of fields.
In fact, one restricts to {\em compactly-supported} fields,
since the action functional is rarely well-defined on all fields when the source manifold is non-compact.
(Note this domain of compactly-supported fields is all one needs for making variational arguments or for constructing a BV quantization.)

An edge corresponds an element $P$ of the 2-fold tensor product of the space of fields,
often called a {\em propagator}.
More precisely, the edge should correspond to
the Green's function for the linear differential operator 
appearing in the associated quadratic term of the action;
hence the propagator is an element of the {\em distributional completion} of the 2-fold tensor product.
For us the $\beta\gamma$ leg should be labeled by $\dbar^{-1} \otimes {\rm id}_V$,
where $\dbar^{-1}$ denotes an inverse to the Dolbeault operator on functions.
The $bc$ leg should be labeled by $\dbar^{-1}_T$, 
the inverse of the Dolbeault operator on the bundle~$T^{1,0}$.

Given a graph~$\Gamma$, one should contract the tensors associated to the vertices and edges.
We denote the linear functional for this graph by~$w_\Gamma(P,I)$,
where $w$ stands for ``weight,'' the term $P$ indicates we label edges by the propagator~$P$,
and the term $I$ indicates we label vertices by the ``interaction'' term of the action~$S$ 
(i.e., the terms that are cubic and higher).

This contraction is not always well-posed, unfortunately.
Each vertex labels a distributional section of some vector bundle on~$\Sigma$,
and each edge labels a distributional section of a vector bundle on~$\Sigma^2$.
Thus the desired contraction can be written {\em formally} as an integral over the product manifold~$\Sigma^{v}$,
where $v$ denotes the number of vertices.
In most situations this contraction is ill-defined, 
since one cannot (usually) pair distributions.
Concretely, one sees that the integral expression is divergent.

Thus, to avoid these divergences, one labels the edges by a smooth replacement of the Green's functions. 
(Imagine replacing a delta function $\delta_0$ by a bump function.)
Since one can pair smooth functions and distributions,
each graph yields a linear functional on fields using these mollified edges.
Thus we have {\em regularized} the divergent expression.

But now this linear functional depends on the choice of mollifications.
Hence the challenge is to show that 
if one picks a sequence of smooth replacements that approaches the Green's function,
there is a well-defined limit of the linear functionals.

\subsubsection{}

We will now sketch one method well-suited to complex geometry
that allows us to see that no divergences appear for the holomorphic bosonic string.
Our approach is an example of the renormalization method developed by Costello in ~\cite{CosBook},
which applies to many more situations.

Our primary setting in this section is $\Sigma=\CC$.
For this Riemann surface, 
a standard choice of Green's function for the $\dbar$ that acts on functions is
\[
P(z,w) = \frac{1}{2 \pi i} \frac{\d z + \d w}{z-w}.
\]
It is a distributional one-form on $\CC^2$ that satisfies $\dbar \otimes 1(P) = \delta_\Delta$, 
where $\delta_\Delta$ is the delta-current supported along the diagonal $\Delta: \CC \hookrightarrow \CC^2$ and providing the integral kernel for the identity.
In terms of our discussion above,
we view this one-form as a distributional section of the fields $\gamma$ and~$\beta$: 
for example, for fixed $w$, the one-form $\d z/(z - w)$ is a $\beta$ field in the $z$-variable 
as it is a $(1,0)$-form.
(This propagator is for the $\beta\gamma$ fields---and one must tensor with a kernel for the identity on $V$---but a similar formula provides a propagator for the $bc$ fields.)

\subsubsection{}

We will now describe the integral associated to a simple diagram.
For simplicity, we assume $V = \CC$ so that the $\gamma$ and $\beta$ fields are simply functions and $1$-forms on $\CC$, respectively.
Consider a ``tadpole'' diagram, Figure \ref{fig:tadpole}, $\Gamma_{\rm tad}$ whose outer legs are $c$~fields 
(i.e., vector fields on~$\CC$).

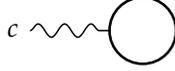
\begin{figure}
\begin{tikzpicture}[line width=.2mm, scale=1.5]
		\draw[vector](180:1) -- (180:.3cm);
			\node at (180:1.15) {$c$};
		\draw[fill=black] (0,0) circle (.3cm);
		\draw[fill=white] (0,0) circle (.29cm);
	    	\clip (0,0) circle (.3cm);
\end{tikzpicture}
\caption{The tadpole diagram $\Gamma_{\rm tad}$}
\label{fig:tadpole}
\end{figure}

There is only one vertex here, corresponding to the cubic function on fields
\[
w_{\Gamma_{\rm tad}}(P,I_{\rm string}) = \int_{z \in \CC} \beta \wedge c\gamma.
\]
If the field $c$ is of the form $f(z) \d \zbar \partial_z$,
with $f$ compactly supported, 
then our integral is
\[
\int_{z \in \CC} \beta \wedge f(z)(\partial_z\gamma) \d \zbar.
\]
(Note that a general cubic function could be described as an integral over $\CC^3$,
but our function is supported on the small diagonal $\CC \hookrightarrow \CC^3$.)
The linear functional for this tapole diagram should be given by inserting the propagator $P$ in place of the $\beta$ and $\gamma$ fields. 
Hence it ought to be given by the following integral over~$\CC$:
\[
\int_{z \in \CC} c(z)P(z,w)|_{z = w}  
= \int_{z \in \CC} f(z) \partial_z \left(  \frac{1}{2 \pi i} \frac{\d z + \d w}{z-w}\right)|_{z = w}\, \d \zbar.
\]
This putative integral is manifestly ill-defined,
since the distribution is singular along the diagonal.

\subsubsection{}

We smooth out the propagator $P$ using familiar tools from differential geometry.
Fix a Hermitian metric on $\Sigma$, 
which then associates provides an adjoint $\dbar^*$ to the Dolbeault operator~$\dbar$.
For the usual metric on $\CC$, we have
\[
\dbar^* = -2 \frac{\partial}{\partial (\d \zbar)} \frac{\partial}{\partial z}.
\]
In physics one calls a choice of the operator $\dbar^*$ a {\em gauge-fix}.
The commutator $[\dbar,\dbar^*]$, which we will denote $D$, 
is equal to $\tfrac{1}{2} \Delta$, where $\Delta$ is the Laplace-Beltrami operator for this metric.
In the physics literature, explicit gauge fixes for the bosonic string can be found in~\cite{Bochicchio}.

We introduce a smoothed version of the propagator using the heat kernel~$e^{-tD}$,
which is a notation that denotes a solution to the heat equation $\partial_t f(t,z) + D f(t,z) = 0$.
For $\CC$ with the Euclidean metric, the standard heat kernel is
\[
e^{-tD}(z,w) =  \frac{1}{4\pi t} e^{-|z-w|^2/4t} (\d z - \d w) \wedge (\d\zbar - \d\wbar) . 
\]
For $0 < \ell < L < \infty$, we define
\[
P_\ell^L = \dbar^* \int_{\ell}^L e^{-tD}\d t.
\]
We compute
\[
\dbar P_\ell^L = D \int_{\ell}^L e^{-tD}\d t =  \int_{\ell}^L \frac{d}{dt} e^{-tD}\d t = e^{-LD} - e^{-\ell D}.
\]
In the limit as $\ell \to 0$ and $L \to \infty$, the operator $P_\ell^L$ goes to a propagator (or Green's function) $P$ for~$\dbar$.
To see this, consider an eigenfunction $f$ of $D$ where $Df=\lambda f$ where $\lambda$ is a non-negative real number. 
Then
\[
(\dbar P_\ell^L) f = (e^{-L\lambda} - e^{-\ell \lambda})f, 
\]
which goes to $f$ as $L \to \infty$ and $\ell~\to~0$.
Thus, if one works with the correct space of functions, 
$P_\ell^L$ is almost an inverse to $\dbar$;
moreover, it is a smooth function on $\Sigma~\times~\Sigma$. 

\subsubsection{}

We now return to the tadpole diagram and put $P_\ell^L$ on the edge instead of~$P$.
(We again assume $V = \CC$ for simplicity.)
The propagator is
\begin{align}\label{propagator}
P_\ell^L(z,w) &= \int_{\ell}^L \d t \, \frac{\partial}{\partial (\d \zbar)} \frac{\partial}{\partial z}\left( \frac{1}{4\pi t} e^{-|z-w|^2/4t} (\d z - \d w) \wedge (\d\zbar - \d\overline{w})\right)\\
&= \int_{\ell}^L \d t \frac{1}{4\pi t} \frac{\zbar - \overline{w}}{2t} e^{-|z-w|^2/4t} (\d z - \d w).
\end{align}
Note that it is smooth everywhere on~$\CC^2$.
The integral for the tadpole diagram is 
\begin{align*}
w_{\Gamma_{\rm tad}}(P_\ell^L,I_{\rm string})
&= \int_{z \in \CC} c(z)P_\ell^L(z,w)|_{z = w}  \\
&= \int_{z \in \CC} \int_{\ell}^L \d t f(z) \partial_z \left(\frac{1}{4\pi t} \frac{\zbar - \overline{w}}{2t} e^{-|z-w|^2/4t} (\d z - \d w) \right)|_{z = w}\, \d \zbar\\
&= \int_{z \in \CC} \int_{\ell}^L \d t f(z) \left(\frac{1}{4\pi t} \left(\frac{\zbar - \overline{w}}{2t}\right)^2 e^{-|z-w|^2/4t} (\d z - \d w) \right)|_{z = w}\, \d \zbar\\
&= 0,
\end{align*}
since the integrand vanishes along the diagonal.
Note that this integral is independent of $\ell$ and $L$ and hence the limit is zero.

\subsubsection{}

By explicitly analyzing the $\ell \to 0$ limit for the integral associated to every Feynman diagram,
we find the following result.

\begin{prop}\label{prop: no counterterms}
For any graph $\Gamma \in {\bf Graph}_{\rm string}$ allowed by the combinatorics of the string action functional and for any $L > 0$,
there is a well-defined limit $\lim_{\ell \to 0} w_{\Gamma}(P_{\ell}^L,I_{\rm string})$.
\end{prop}

We denote this limit by~$w_{\Gamma}(P_{0}^L,I_{\rm string})$. 
The necessary manipulations and inequalities referenced below are very close to those used in~\cite{wg2, GGW}.

\begin{proof}[Outline of proof]
When $\Gamma$ is a tree, there is never an issue with divergences; 
we could even use the Green's function $\dbar^{-1}$ on each edge.
To see this, note that one can view a tree as having a distinguished root,
given by the leg that is either of $\beta$ or $b$~type.
One can then see the tree as describing a multilinear map from the leaves (i.e., legs that are not roots) to the root.
Indeed, one can view each cubic vertex as such an operator.
For instance, $\langle b, [c,c]\rangle$ corresponds to the Lie bracket of vector fields,
since we view $\langle b,-\rangle$ as an element of the $c$~fields.
For a tree, one can then input arbitrary elements into the leaves, 
apply the operations labeled by the vertices,
apply the operator labeled by the edge, and so on,
until one reaches the root.
The composite multilinear operator sends smooth sections to smooth sections,
even if the edges are labeled by distributional sections,
since the associated operator sends smooth sections to smooth sections.

When $\Gamma$ is a one-loop graph, it consists of a wheel with trees attached to the outer legs.
By the preceding argument, we know those trees do not introduce singularities;
hence any divergences are due solely to the wheel.
It thus suffices to consider pure wheels (i.e., those with no trees attached).

Let the wheel have $n$ vertices. 
The $k$th vertex has a coordinate $z_k$ on $\CC$;
the $k$th external leg has input $c_k = f_k(z_k,\zbar_k) \d\zbar_k\, \partial_{z_k}$, 
where $f_k$ is a compactly-supported smooth function.
Then the integral has the form
\[
\int_{(z_1,\ldots,z_n) \in \CC^n}\d^n \zbar \,(f_1  \partial_{z_1} P_\ell^L(z_1,z_n))(f_2 \partial_{z_2} P_\ell^L(z_2,z_1)) \cdots (f_n\partial_{z_n} P_\ell^L(z_n,z_{n-1})),
\] 
since the $k$th input will act on one of the propagators entering the $k$th vertex.
One needs to show that this expression has a finite $\ell \to 0$ limit.

Let us prove this limit exists for the case $n=2$.
Then we have
\begin{align*}
\int_{z_1,z_2 \in \CC} \d\zbar_1\d\zbar_2 \int_{\ell}^L \d t_1 \int_{\ell}^L \d t_2\, 
& f_1(z_1)f_2(z_2) 
\partial_{z_1} \left(\frac{1}{4\pi t_1} \frac{\zbar_1 - \zbar_2}{2t_1} e^{-|z_1-z_2|^2/4t_1} (\d z_1 - \d z_2) \right)\\
& \times \partial_{z_2} \left(\frac{1}{4\pi t_2} \frac{\zbar_1 - \zbar_2}{2t_2} e^{-|z_1-z_2|^2/4t_2} (\d z_2 - \d z_1) \right),
\end{align*}
which is already a bit lengthy.
As our focus is on showing a limit exists, we will throw out unimportant factors and simplify the expression.
First, note that taking the partial derivative $\partial_{z_i}$ will simply multiply the integrand by $(\zbar_1 - \zbar_2)/2t_i$.
Moreover, we change coordinates to $u = z_1 - z_2$ and $v = z_2$. 
Then the integral is proportional to
\[
\int_{\ell}^L \d t_1 \int_{\ell}^L \d t_2\int_{\CC^2} \d^2 u \, \d^2 v \, f_1 f_2 \frac{\overline{u}^4}{t_1^3 t_2^3} e^{-|u|^2(\tfrac{1}{t_1} + \tfrac{1}{t_2})}.
\]
We take the integral over $v$ last;
it will be manifestly well-behaved after we take the other integrals.

Thus consider the integral just over $u \in \CC$,
so that we are computing the expected value of $F=f_1 f_2$ against a Gaussian measure 
whose variance is determined by $t_1$ and~$t_2$.
(Namely, the variance is~$t_1 t_2/(t_1+t_2)$.)
We might as well focus on values of $t_i$ that are very small, 
as those would be the source of divergences when $\ell \to 0$.
For small $t_i$, we only care about the behavior of $F$ near the origin as the measure is concentrated near the origin.
Thus, consider a partial Taylor expansion of $F$.
The polynomial part can be computed quickly since the expected values of monomials against a Gaussian measure (i.e., the moments) have a simply expression in terms of the variance.
The first nonzero contribution would come from the $u^4$ term in the Taylor expansion of $F$,
and it contributes a factor of the form $(t_1 t_2/(t_1+t_2))^5$,
up to constant that we ignore.
We are left with
\[
\int_{\ell}^L \d t_1 \int_{\ell}^L \d t_2 \frac{(t_1t_2)^{3}}{(t_1 + t_2)^{5}} 
\leq \int_{\ell}^L \d t_1 \int_{\ell}^L \d t_2 \, 2^{-5}\sqrt{t_1t_2} 
= 2^{-5}(L^{3/2} - {\ell}^{3/2})^2,
\]
where we use the arithmetic-geometric mean inequality $\sqrt{t_1t_2}/(t_1+t_2)\leq 1/2$ in the middle.
This expression has a finite limit as $\ell \to 0$.
The higher terms in the Taylor expansion contribute bigger powers of the variance 
and hence have $\ell \to 0$ limits.
Finally, the expected value of the error term of our partial Taylor expansion, 
which vanishes to some positive order at the origin,
can be bounded in such as way that an $\ell \to 0$ limit exists.
\end{proof}

We can now define the effective theory that we consider for the string. 

\begin{dfn}
The {\em renormalized action functional} at scale~$L$ for the holomorphic bosonic string is
\[
I[L] = \sum_{\Gamma \in {\bf Graph}_{\rm string}^{(0)}} w_{\Gamma}(P_{0}^L,I_{\rm string}) + \hbar\sum_{\Gamma \in {\bf Graph}_{\rm string}^{(1)}} w_{\Gamma}(P_{0}^L,I_{\rm string}).
\]
We denote the first summand---the tree-level expansion---by $S_0[L]$ 
and the second summand---the one-loop expansion---by~$S_1[L]$.
We use the notation $S[L] = S_{free} + I[L]$ where $S_{free}$ is the classical free part of the action functional. 
\end{dfn}

\begin{rmk} For any functional $J$, let $w(P_{\ell}^L, J)$ denote the sum over all graphs as above with the smooth propagator $P_{\ell}^L$ placed at the edges and $J$ placed at the vertices. 
Then, the family $\{I[L]\}$ satisfies the {\em homotopy RG equation}
\ben
I[L] = w(P_\ell^L , I[\ell]).
\een
The operator $w(P_\ell^L,-)$ defines a homotopy equivalence between the theory at scale $\ell$, defined using $S[\ell]$, and the theory at scale $S$, defined using $S[L]$. 
\end{rmk}

\subsection{The quantum master equation}
\label{subsec: QME}

In the BV formalism the basic idea is to replace integration against a path integral measure $e^{-S(\phi)/\hbar} \cD \phi$ with a cochain complex.
In this cochain complex, we view a cocycle as defining an observable of the theory,
and its cohomology class is viewed as its expected value against the path integral measure.
For toy models of finite-dimensional integration, see \cite{GJF};
these examples are always cryptomorphically equivalent to a de Rham complex,
which is a familiar homological approach to integration.

Hence the content of the path integral, in this approach, is encoded in the differential. 
A key idea is that the differential is supposed to behave like a divergence operator for a volume form:
recall that given a volume form $\mu$ on a manifold, 
its divergence operator maps vector fields to functions by the relationship
\[
{\rm div}_\mu({\mathcal X}) \mu = L_{\mathcal X} \mu.
\] 
This relationship, in conjunction with Stokes lemma, 
implies that if a function~$f$ is a divergence ${\rm div}_\mu({\mathcal X})$,
then $\int f \mu = 0$,
i.e., its expected value against the measure~$\mu$ is zero.
The BV formalism axiomatizes general properties of divergence operators;
a putative differential must satisfy these properties to provide a BV quantization.

When following the algorithm of Section~\ref{sec:bvoverview},
we want the renormalized action
\[
S = S^{\rm cl} + \hbar S_1 + \hbar^2 S_2 + \cdots
\]
to determine a putative differential $\d^q_S$ on the graded vector space of observables.
To explain this operator, we need to describe further algebraic properties on the observables
that the BV formalism uses.

First, in practice, the observables are the symmetric algebra generated by the continuous linear duals to the vector spaces of fields.
There is also a pairing on fields that is part of the data of the classical BV theory,
between each field and its ``anti-field.''
(This pairing is a version of the action of constant vector fields on functions in the toy models.)
In our case, there is the pairing between $b$ and $c$ and between $\beta$ and~$\gamma$, respectively.
It behaves like a ``shifted symplectic'' pairing as it has cohomological degree~$-1$,
and hence it determines a degree~1 Poisson bracket $\{-,-\}$ on the graded algebra of observables.
Finally, the pairing also determines a second-order differential operator $\Delta_{BV}$ on the algebra of observables by the condition that
\[
\Delta_{BV}(FG) = (\Delta_{BV}F)G + (-1)^F F(\Delta_{BV}G) + \{F,G\}.
\]
(This equation is a characteristic feature of divergence operators with respect to the product of polyvector fields.)

With these structures in hand, we can give the formula
\[
\d^q_S=\{S,-\} + \hbar\Delta_{BV}
\]
for the putative differential.
As $S$ has cohomological degree~0, the operator $\{S,-\}$ has degree~1.
We remark that modulo~$\hbar$, one recovers the differential $\{S^{\rm cl},-\}$ on the classical observables;
the zeroth cohomology of the classical observables is functions on the critical locus of the classical action~$S^{\rm cl}$.

By construction, this putative differential $\d_S^q$ satisfies the conditions of behaving like a divergence operator.
The only remaining condition to check is that it is square-zero.
This condition ends up being equivalent to $S$ satisfying the {\em quantum master equation}
\begin{equation}
\hbar \Delta_{BV} S + \frac{1}{2}\{S,S\} = 0.
\end{equation}
More accurately, $\d^q_S$ is a differential if and only if the right hand side is a constant.

\subsubsection{}

We now turn to examining this condition in our setting.
It helps to understand it is diagrammatic terms.

As the bracket is determined by a linear pairing,
it admits a simple diagrammatic description as an edge.
For instance, given an observable $F$ that is a homogeneous polynomial of arity~$m$
and an observable $G$ of arity~$n$, 
then $\{F,G\}$ has arity~$m+n-2$.
It can be expressed as a Feynman diagram 
where the edge connecting $F$ and $G$ is labeled by a 2-fold tensor~$K$.

The BV Laplacian acts by attaching an edge labeled by~$K$ as a loop in all possible ways.
This diagrammatic behavior corresponds to the fact that $\Delta_{BV}$ is a constant-coefficient second-order differential operator.

The tensor~$K$ determined by the pairing on fields is distributional.
As one might expect from our discussion of divergences above,
these diagrammatic descriptions of the BV bracket and Laplacian are thus typically ill-defined.
In other words, the quantum master equation is {\em a priori} ill-posed for the same reason that the initial Feynman diagrams are ill-defined.
We can apply, however, the same cure of mollification.

\subsubsection{}

Costello's framework \cite{CosBook} provides an approach to renormalization built to be compatible with the BV formalism.
A key feature is that for each ``length scale''~$L>0$, 
there is a BV bracket $\{-,-\}_L$ and BV Laplacian $\Delta_L$.
The scale~$L$ renormalized action $S[L]$ satisfies the scale~$L$ quantum master equation~(QME)
\[
\hbar \Delta_{L} S[L] + \frac{1}{2}\{S[L],S[L]\}_L = 0
\]
if and only if $S[L']$ satisfies the scale~$L'$ quantum master equation for every other scale~$L'$, see Lemma 9.2.2 in \cite{CosBook}.
Hence, we say a renormalized action satisfies the quantum master equation if its solves the scale~$L$ equation for some~$L$.

Thus it remains for us to describe the scale~$L$ bracket and BV Laplacian in our setting,
so that we can examine whether the renormalized action satisfies the quantum master equation.

\begin{dfn}
The {\em scale~$L$ bracket} $\{-,-\}_L$ is given by pairing with the scale~$L$ heat kernel
\[
K_L(z,w) = \frac{1}{4\pi L} e^{-|z-w|/4L} (\d z - \d w) \wedge (\d\zbar - \d\overline{w}). 
\]
The {\em scale~$L$ BV Laplacian} $\Delta_L$ is given by the contraction~$\partial_{K_L}$.
\end{dfn}

These definitions mean that testing the quantum master equation leads to diagrams whose integrals are similar to those we encountered earlier.
We explain the diagrammatics and sketch the relevant integrals in the proof of the following result,
which characterizes when the string action admits a BV quantization.

We emphasize that up to now, we have not indicated explicitly which vector space $V$ is the target space for our string.
But the action functional explicitly depends on this choice,
so here we will write $S_V$ for the action with target~$V$.

\begin{prop} \label{prop anomaly}
The obstruction to satisfying the quantum master equation is the functional
\[
Ob_V[L] = \hbar \Delta_{L} S_V[L] + \frac{1}{2}\{S_V[L],S_V[L]\}_L.
\]
It has the form
\[
Ob_V[L] = \hbar (\dim_\CC(V) - 13) F[L],
\]
where $F[L]$ is a functional independent of~$V$.
\end{prop}

In short, the failure to satisfy the QME is a linear function of the dimension of the target space~$V$.
In particular, when $V \cong \CC^{13}$, 
the obstruction vanishes and the renormalized action {\em does} satisfy the QME, 
giving us an immediate corollary.
(Note that we do {\em not} need to know $F[L]$ to recognize that the obstruction vanishes!)

\begin{cor}
When the target vector space is 13-dimensional (i.e., has 26 real dimensions),
the holomorphic bosonic string admits a BV quantization.
\end{cor}

\begin{proof}
It is a general feature of Costello's formalism that the tree-level term $S_0[L]$ of the renormalized action satisfies the scale~$L$ equation
\[
\{S_0[L],S_0[L]\}_L = 0,
\]
known as the classical master equation.
Hence the first obstruction to satisfying the QME can only appear with positive powers of~$\hbar$.
We can also see quickly that no terms of~$\hbar^2$ appear:
the one-loop term $S_1[L]$ is only a function of the $c$~field, 
so 
\[
\{S_1[L],S_1[L]\}_L = 0 \quad\text{and}\quad \Delta_L S_1[L] = 0.
\]
Hence the obstruction to satisfying the QME is precisely
\[
\hbar\left( \{S_0[L],S_1[L]\} + \Delta_L S_0[L] \right).
\]
Thus we see that the obstruction is a multiple of~$\hbar$.
For simplicity, we will divide out that factor and let $Ob_V$ denote the term inside the parenthesis.

Consider the term $\{S_0[L],S_1[L]\}_L$. 
Diagrammatically, it corresponds to attaching a tree with a $b$ ``root'' to a wheel using an edge labeled by~$K_L$.
Arguments similar to Lemma 16.0.3 of \cite{wg2} carry over to account for the vanishing of this term in the $L \to 0$ limit. 

Now consider the term $\Delta_L S_0[L]$. 
Diagrammatically, it corresponds to turning a tree into a wheel by using an edge---labeled by $K_L$---to attach the root to an incoming leaf.
There are thus two kinds of wheels that appear, since there are two kinds of trees.
There are the wheels where the $K$ edge is for $bc$ fields.
Note that these wheels are the same for every choice of target $V$
as they only depend on the $bc$ fields, i.e., are independent of the $\beta\gamma$ fields.
These will contribute a term $F[L]$ to the obstruction.
On the other hand, there are the wheels where the $K$ edge is for $\beta\gamma$ fields.
These depend on $V$ but in a very simple way: 
the distribution $K$ is just the heat kernel tensored with the identity on~$V$, 
and hence the contraction amounts to taking $\dim_\CC(V)$ copies of the $V = \CC$ value.
In other words, the $\beta\gamma$ wheels contribute a term $\dim_\CC(V) G[L]$ to the obstruction,
where $G[L]$ is the value for $V = \CC$.
The last part of the proof of the theorem is a direct calculation of the functionals $F[L]$ and $G[L]$. 
So as to not diverge from our track of thought we include this calculation in Appendix \ref{sec:calculation} where we show that $F[L], G[L]$ are both independent of $L$ and satisfy $F = -13 G$, thus completing the proof.
\end{proof}

\begin{rmk} One can consider coupling the $\beta\gamma$ system to any tensor bundle on the Riemann surface. 
For instance, suppose $\gamma$ is a section of $T_\Sigma^{\tensor n}$ and hence $\beta$ is a section of $T^{* \otimes {n+1}}_\Sigma$.
In this case, one can show that the part of the obstruction with internal edges labeled by the $\beta\gamma$ propagators contributes a factor $(6n^2 + 6n + 1)G$, with $G$ the same functional above. 
\end{rmk}

\section{OPE and the string vertex algebra}

Vertex algebras are mathematical objects that axiomatize the behavior of local observables 
(i.e., point-like observables) of a chiral conformal field theory (CFT),
such as the $bc\beta\gamma$ system or the holomorphic bosonic string.
In particular, the operator product expansion (OPE) for these local observables---which is of paramount importance in understanding a chiral CFT---is encoded by the vertex operator of the vertex algebra of the CFT.
(We will not review vertex algebras here
as there are many nice expositions~\cite{FHL, BZF}.)

In this section we will explain how to extract the vertex algebra of the holomorphic bosonic string,
using machinery developed in \cite{CG1,LiVA,GGW}.
The answer we recover is precisely the chiral sector of the usual bosonic string.

\subsection{Some context}

In the BV formalism one constructs a cochain complex of observables,
for both the classical and the quantized theory, if it exists.
The cochain complexes are local on the source manifold of a theory:
on each open set $U$ in that manifold~$\Sigma$,
one can pick out the observables with support in~$U$ by asking for the observables that vanish on fields with support outside~$U$.
Furthermore, you can combine observables that have support on disjoint open sets.
It is the central result of~\cite{CG1,CG2} that the observables also satisfy a local-to-global property,
akin to the sheaf gluing axiom.
Such a structure is known as a {\em factorization algebra} on~$\Sigma$.

We will not need that general notion here.
Instead, we will use vertex algebras.
Theorem~5.2.3.1 of~\cite{CG1} explains how a factorization algebra~$F$ on $\Sigma = \CC$
yields a vertex algebra~$\Vert(F)$, under natural hypotheses on~$F$. 
It assures us that the observables of a chiral CFT determine a vertex algebra.

In particular, Section~5.3 of~\cite{CG1} examines the free $\beta\gamma$ system in great detail.
Its main result is that the well-known $\beta\gamma$ vertex algebra $\cV_{\beta\gamma}$ is recovered by the two-step process of BV quantization, which yields a factorization algebra, and then the extraction of a vertex algebra.
The exact same arguments apply to the free $bc$ system,
recovering the vertex algebra~$\cV_{bc}$;
and of course, the exact same arguments apply to the free $bc\beta\gamma$ system.

Let $V$ denote the vector space appearing in the $\beta\gamma$ contribution of the holomorphic bosonic string theory, as introduced in Section~\ref{sec:classical}.
Let $\Obs^\q_{free}$ denote the observables of this theory on $\Sigma = \CC$.
As a quantization of a free field theory, it is a factorization algebra valued in the category of $\CC [\hbar]$-modules.
In particular, the associated vertex algebra $\Vert(\Obs^\q_{free})$ is also valued in $\CC[\hbar]$-modules.
Putting the claims together, we have the following.

\begin{prop}\label{prop: bcbg vertex}
For $n = \dim_{\CC}(V)$, 
there is an isomorphism of vertex algebras
\ben
\Vert(\Obs^{\q}_{free})_{\hbar = 2 \pi i} \cong \cV_{bc} \tensor \cV_{\beta\gamma}^{\tensor n} 
\een 
where on the left-hand side we have set~$\hbar = 2\pi i$.
\end{prop}

\subsection{A reminder on the chiral algebra of the string}\label{subsec: string vert}

We now provide a brief review of the vertex algebra for the chiral sector of the bosonic string. 
For a detailed reference we refer the reader to \cite{LZ1,LZ2}. 
The construction builds a {\em differential graded vertex algebra}, 
which is simply a vertex algebra in the category of cochain complexes. 
The underlying graded vertex algebra has a state space of the form
\ben
\cV_{\beta \gamma}^{\tensor 13} \tensor \cV_{bc},
\een
where $\cV_{\beta\gamma}$ and $\cV_{bc}$ are the $\beta\gamma$ and $bc$ vertex algebras, respectively. 
The $\beta$ and $\gamma$ generators are in grading degree zero, the $c$ generator is in grading degree~$-1$, and the $b$ is in grading degree~$1$. 
In the physics literature it is referred to as the {\em BRST} grading or {\em ghost number}.

Forgetting the cohomological (or BRST) grading, this vertex algebra is a conformal vertex algebra of central charge zero (by construction). 
In particular, this means that the vertex algebra has a stress-energy tensor. 
Explicitly, it is of the form
\ben
T_{\rm string} (z) = \left(\sum_{i = 1}^{13} \beta_i (z) \partial_z \gamma_i (z) + \partial_z \beta_i(z) \gamma_i (z) \right) + \left(b(z) \partial_z c(z) + 2 \partial_z b(z) c(z) \right) . 
\een
Note that $T_{\rm string}$ is of cohomological degree zero. 
The first parenthesis is interpreted as the stress-energy tensor of the vertex algebra $\cV_{\beta \gamma}^{\tensor 13}$ and the second term is the stress-energy tensor of $\cV_{bc}$. 

We have not yet described the differential on the graded vertex algebra. 
The BRST differential is defined to be the vertex algebra derivation obtained by taking the following residue
\be\label{brst}
Q^{BRST} = \oint c(z) T_{\rm string}(z) .
\ee
By construction this operator satisfies $(Q^{BRST})^2 = 0$. 

\begin{dfn} The {\em string vertex algebra} is the dg vertex algebra 
\ben 
\cV_{\rm string} = \left(\cV_{\beta \gamma}^{\tensor 13} \tensor \cV_{bc}, \; Q^{BRST}\right)  .
\een
\end{dfn}

There is another grading on $\cV_{\rm string}$ coming from the eigenvalues of the vertex algebra derivation $c_0$ called the {\em conformal dimension}. 
In particular, this determines a filtration and we can consider the associated graded ${\rm Gr} \; \cV_{\rm string}$. 
The conformal weight grading preserves the cohomological grading so that this object still has the structure of a dg vertex algebra. 

Note that the cohomology of a dg vertex algebra is an ordinary (graded) vertex algebra. 
The cohomology of the string vertex algebra is called the {\em BRST cohomology} of the bosonic string. 
In the remainder of this section we will show how we recover the string vertex algebra from the quantization of the holomorphic bosonic string.

\subsection{The case of the string}

The holomorphic bosonic string is a chiral CFT and so the machinery of~\cite{CG1} applies to it.
One can extract a vertex algebra directly by this method, as one does with the free $bc\beta\gamma$ theory.

But there is a slicker approach, using Li's work~\cite{LiVA},
which studies chiral deformations of {\em free} chiral BV theories such as the free $bc\beta\gamma$ system.
Recall that a deformation of a classical field theory is given by a local functional. 
We have seen that this is essentially the data of a Lagrangian density, which is a density valued multilinear functional that depends on (arbitrarily high order) jets of the fields. 
In other words, for a field $\varphi$, a Lagrangian density is of the form
\ben
\cL(\varphi) = \sum (D_{k_1} \varphi) \cdots (D_{k_m} \varphi) \cdot {\rm vol}_\Sigma
\een 
for $C^\infty(\Sigma)$-valued differential operators $D_{k_i}$.
By a {\em chiral} Lagrangian density we mean a Lagrangian for which the differential operators $D_{k_i}$ are all holomorphic. 
For instance, on $\Sigma = \CC$, we require $D_{k_i}$ to be a sum of operators of the form $f(z) \partial_z^n$ where $f(z)$ is a holomorphic function. 
On $\Sigma = \CC$ we will also require the chiral Lagrangian to be translation invariant. 
This means that all differential operators $D_{k_i}$ are of the form $\partial_z^n$. 
Thus, a {\em translation-invariant chiral deformation} is a local functional of the form
\ben
I(\varphi) = \sum \int (\partial^{k_1}_z \varphi) \cdots (\partial^{k_m} \varphi) \d^2 z .
\een
Such a deformation stays within the class of chiral CFTs.

One of Li's main results is that for a free chiral BV theory with action $S_{\rm free}$ and associated vertex algebra $\cV_{\rm free}$, one has the following:
\begin{itemize}
\item For any chiral interaction~$I$, the action $S = S_{\rm free} + I$ yields a renormalized action functional $I[L] = \lim_{\ell \to 0} W(P_\ell^L, I)$
that requires no counterterms. 
That is, the weights of all Feynman diagrams are finite (compare to Proposition \ref{prop: no counterterms}),
\item If the renormalized action $\{I[L]\}$ satisfies the quantum master equation, 
then it determines a vertex algebra derivation $D_I$ of~$\cV_{\rm free}$ of the form
\ben
D_I = \oint I^q\, \d z
\een
that is square-zero and of cohomological degree one.
Here, $I^q = \lim_{L \to 0} I[L]$, where $I[L]$ is the renormalized action functional.
Modulo~$\hbar$, it agrees with the chiral interaction $I$, but it has $\hbar$-dependent terms that provide the ``quantum corrections'' to the classical action.
\item The dg vertex algebra $\cV_I$ for such an action $\{I[L]\}$ has the same underlying graded vertex algebra $\cV_{\rm free}$ but it is equipped with the differential $\oint I^q \d z$. 
\end{itemize}
This construction significantly reduces the work of constructing the vertex algebra for the chiral deformation, as one need not analyze the factorization algebra directly.

\begin{rmk} The fact that $I$ satisfies the quantum master equation implies that one has a map, for each open set $U \subset \CC$, from the free factorization algebra evaluated on $U$ to the factorization algebra of the deformed theory evaluated on $U$:
\ben
e^{I /\hbar} : \Obs^q_{free}(U) \to \Obs^q_I (U) .
\een
This map sends an observable $O \in \Obs^q_{free}(U)$ to $O \cdot e^{I/\hbar}$. 
In fact, this map is an isomorphism with inverse given by $O \mapsto O \cdot e^{-I/\hbar}$. 
So, open by open, the factorization algebra assigns the same vector space for the deformed theory.
This isomorphism is {\em not} compatible with the factorization product, so we do get a different factorization algebra in the presence of a deformation.
\end{rmk}

The holomorphic bosonic string with target $V=\CC^{13}$ provides a concrete example of this situation.
The free theory is the $bc\beta\gamma$ system,
the holomorphic bosonic string is a chiral deformation of it, 
and we have seen that the renormalized action of the string satisfies the QME.
Hence we obtain the following.

\begin{prop} 
\label{prop: fact is vert}
Let $\Obs^\q_{\rm string}$ be the factorization algebra on $\Sigma = \CC$ of the holomorphic bosonic string with target~$V = \CC^{13}$. 
Let $\Vert(\Obs^q_{\rm string})$ be the dg vertex algebra (defined over $\CC[\hbar]$) obtained via Li's construction. 
There is an isomorphism  
\[
\cV_{\rm string} \cong \Vert(\Obs^q_{\rm string})\big|_{\hbar = 2 \pi i}
\]
of dg vertex algebras.
Moreover, this vertex algebra is isomorphic to the chiral sector of the bosonic string as in Section~\ref{subsec: string vert}.
\end{prop}

The factorization algebra $\Obs^\q_{\rm string}$ is also a quantization of the factorization algebra $\Obs^{\cl}_{\rm string}$ of classical observables.
We have noted that the classical observables of any theory has the structure of a $P_0$ factorization algebra, and the $\hbar \to 0$ limit of $\Obs^\q_{\rm string}$ is isomorphic to $\Obs^{\cl}_{\rm string}$ as $P_0$ factorization algebras.
By definition, the classical observables are simply functions on the solutions to the classical equations of motion.
The $P_0$ structure is induced from the symplectic pairing of degree -1 on the fields. 
The classical factorization algebra still has enough structure to determine a vertex algebra $\Vert(\Obs^\cl_{\rm string})$.
Moreover, the $P_0$ bracket on the classical observables determines the structure of a {\em Poisson vertex algebra} on $\Vert(\Obs^{\cl}_{\rm string})$. 

\begin{cor} In the classical limit, there is an isomorphism 
\[
\Vert(\Obs^{\cl}_{\rm string}) \cong {\rm Gr} \; \cV_{\rm string}
\]
of Poisson vertex algebras.
\end{cor}

\begin{proof}[Proof of Proposition \ref{prop: fact is vert}] By Proposition \ref{prop: bcbg vertex} we know that the vertex algebra of the associated free theory is identified with the $bc\beta\gamma$ vertex algebra. 
The thing we need to check is that the differential induced from the quantization of the holomorphic string agrees with the differential of the string vertex algebra. 
In fact, we observe that the induced differential $\oint I \,\d z$ from the classical interaction of the holomorphic bosonic string agrees with the BRST charge in Equation (\ref{brst}). 
To see that this persists at the quantum level we need to check that there are no quantum corrections. 
Indeed, this follows from the fact that the quantum master equation holds identically (as opposed to holding up to an exact term in the deformation complex) provided $\dim_\CC V = 13$. 
\end{proof}

\subsection{The $E_2$ algebra and descent}

In this section we highlight a remarkable feature of the vertex algebra associated to the bosonic string. 
At first glance, the theory we have constructed is far from being topological.
Indeed, the classical theory depends delicately on the complex structure of the two-dimensional source. 
Nevertheless, the local observables of the bosonic string behave like the observables of a {\em topological} field theory (TFT). 
In particular, as noted perhaps first by \cite{Getzler}, the observables of a 2-dimensional TFT have the structure of a {\em Gerstenhaber algebra}.
In this section we provide two equivalent methods for extracting this algebra.
The first is intuitive from the point of view of factorization algebras, but has the disadvantage of not giving a concrete description of the algebra. 
The second approach gives an explicit formula for the bracket and is based on the formalism of ``descent" for local operators. 

\subsubsection{The $E_2$ algebra}

We continue to consider the theory on the Riemann surface $\Sigma = \CC$. 
In this section we show how to produce, from the point of view of factorization algebras, the structure of a Gerstenhaber algebra on the BRST cohomology of the bosonic string. 

Recall that a Gerstenhaber algebra is equivalent to an algebra over the operad given by the homology of the little 2-disk operad.
Hence, our approach is to see why the factorization algebra naturally exhibits the structure of a algebra of little 2-disks.
Here we use an important result of Lurie (namely Theorem 5.4.5.9 of~\cite{Lurie}): 
a {\em locally constant} factorization algebra on $\RR^n$ is equivalent to an algebra over the little $n$-disks operad, i.e., an $E_n$-algebra. 

\begin{prop} 
\label{prop: obs is e2}
The factorization algebra $\Obs^\q_{\rm string}$ is locally constant, 
and hence it determines an $E_2$ algebra.
\end{prop}

In particular, the cohomology $H^*(\Obs^\q_{\rm string})$ is an algebra over the cohomology of the $E_2$ operad and hence a Gerstenhaber algebra.

\begin{rmk}
When a topological field theory arises from an action functional (e.g., Chern-Simons theories),
the factorization algebra is locally constant.
Hence such a TFT in $n$ real dimensions produces an $E_n$-algebra, by Lurie's result. 
(This claim holds true, at least, for all the examples with which we are familiar.)
In this sense, holomorphic bosonic string theory is a 2-dimensional topological field theory. 
Moreover, by work of Scheimbauer \cite{Scheim},
every $E_n$ algebra determines a fully-extended framed n-dimensional TFT in the functorial sense, albeit with values in an unusual target $(\infty,n)$-category.
In this sense, at least, the holomorphic bosonic string determines a functorial 2-dimensional TFT.
\end{rmk}

\begin{proof} 
We need to show that for any inclusion of open disks $D \hookrightarrow D'$, the natural map
\ben
\Obs^\q_{\rm string}(D) \to \Obs^\q_{\rm string}(D')
\een
is a quasi-isomorphism. 

We first show that the classical observables are locally constant. 
We have already mentioned that the classical observables are the commutative algebra of functions on the space of solutions to the classical equations of motion. 
This space of solutions forms a sheaf on $\Sigma$, 
since satisfying a PDE is a local condition.
We find it convenient to encode the equations of motion as the Maurer-Cartan equation of a sheaf  of dg Lie algebras:
\ben
\Omega^{0,*}(\Sigma ; \cT_\Sigma) \ltimes \left(\Omega^{0,*}(\Sigma; V)[-1] \oplus \Omega^{1,*}(\Sigma;V^*)[-1] \oplus \Omega^{1,*}(\Sigma ; \cT_\Sigma^*)[-2] \right) . 
\een
(Note that the underlying graded space is simply the fields shifted up by one degree,
which is a generic phenomenon in the BV formalism.)
The dg Lie algebra $\Omega^{0,*}(\Sigma ; \cT_\Sigma)$ is simply a sheaf-theoretic resolution of holomorphic vector fields, with the usual Lie bracket.
Our large dg Lie algebra is a square-zero extension of $\Omega^{0,*}(\Sigma ; \cT_\Sigma)$, 
by the dg module inside the parentheses.
The vector fields act by the Lie derivative on the space
\[
\Omega^{0,*}(\Sigma; V)[-1] \oplus \Omega^{1,*}(\Sigma;V^*)[-1] \oplus \Omega^{1,*}(\Sigma ; \cT_\Sigma^*)[-2],
\]
which is simply a copy of the $\beta\gamma$ system with target vector space~$V$, plus the $b$-field part of the classical theory.
 
For simplicity, we write $\cL = \Omega^{0,*}(\Sigma ; \cT_\Sigma)$ and write $\cM$ for the module inside the parentheses.
In this language, the space of classical observables supported on an open set $U \subset \Sigma$ is the Chevalley-Eilenberg cochain complex
\ben
\Obs^{\cl}_{\rm string}(U) = \clie^*\left(\cL(U) \ltimes \cM(U)\right) = \clie^*\left(\cL(U) ; \; \Sym(\cM(U)^*[-1])\right),
\een
where $\cM(U)^*$ denotes the continuous linear dual of~$\cM(U)$. 

Consider now the case that the open set is a disk $U = D$, which we can assume is centered at zero. 
By the $\dbar$-Poincar\'{e} lemma 
there is a quasi-isomorphism of dg Lie algebras $\cT^{hol}(D) \hookrightarrow \cL(D)$ where $\cT^{\rm hol}(D)$ is the vector space of holomorphic vector fields on $D$. 
Thus, we have a quasi-isomorphism
\[
\clie^*\left(\cT_{hol}(D) ; \; \Sym(\cM(D)^*[-1])\right) \simeq \Obs^{\cl}_{\rm string}(D).
\]
This quasi-isomorphism clearly holds for any disks (and is compatible with inclusions of disks), so it suffices to check that the left-hand side is a quasi-isomorphism for an inclusion of disks.

Consider the composition of Lie algebras
\ben
{\rm W}_1^{\rm poly} \hookrightarrow \cT_{hol} (D) \to {\rm W}_1
\een
where ${\rm W}_1^{\rm poly}$ are the holomorphic vector fields with {\em polynomial} coefficients, and ${\rm W}_1$ is the Lie algebra with {\em power series} coefficients (i.e., formal vector fields).
The second map is the power series expansion, at zero, of a holomorphic vector field. 
We will compare Lie algebra cohomology using these different Lie algebras.

Let $\cA(D)$ denote $\Sym(\cM(D)^*[-1])$.
It determines a module over ${\rm W}_1^{\rm poly}$ by restriction,
which we will abusively denote $\cA(D)$ as well.
Likewise, if $j_0^\infty \cM$ denotes the infinite jet of the sheaf $\cM$ at the origin of the disk $D$,
then it determines a natural module over ${\rm W}_1$.
Then $\Sym(\cM(D)^*[-1])$ determines a ${\rm W}_1$-module that we will also abusively denote by~$\cA(D)$.
 
The inclusion $D \hookrightarrow D'$ then yields a commutative diagram
\ben
\xymatrix{
\clie^*\left({\rm W}_1^{\rm poly} ; \cA(D)\right) \ar[d] & \ar[l] \clie^*\left(\cT_{\rm hol} (D) ; \cA(D))\right) \ar[d] & \ar[l]  \clie^*\left({\rm W}_1 ; \cA(D)\right) \ar[d] \\
\clie^*\left({\rm W}_1^{\rm poly} ; \cA(D')\right) & \ar[l] \clie^*\left(\cT_{\rm hol} (D') ; \cA(D')\right) & \ar[l] \clie^*\left({\rm W}_1 ; \cA(D')\right) .
}
\een
By Lemma \ref{lem: gf} (and an analogous result for polynomial vector fields),
these complexes $\clie^*({\rm W}_1 ; \cM)$ and $\clie^*({\rm W}^{\rm poly}_1; \cM)$ are quasi-isomorphic to the subcomplex consisting of conformal dimension zero elements, 
i.e., to the constants. 
As the conformal dimension zero subcomplex does not depend on the size of the disk, we conclude that vertical arrows on the outside of the commutative diagram are quasi-isomorphisms. 
It follows that the middle vertical arrow is as well, 
thus showing that $\Obs^{\cl}_{\rm string}(D) \to \Obs^{\cl}_{\rm string}(D')$ is a quasi-isomorphism, as desired. 

To finish the proof, we need to prove the quasi-isomorphism for {\em quantum} observables.
Consider the spectral sequence induced from the filtration of the module~$\Sym \;\cM(D)$ by symmetric polynomial degree. 
The $E_1$ page of this spectral sequence is the classical observables above, 
and it converges to the cohomology of the quantum observables. 
As the map of factorization algebras induced by the inclusion $D \hookrightarrow D'$ preserves this filtration, 
we obtain a map of spectral sequences,
which is quasi-isomorphism on the first page.
Hence, $\Obs^{\q}_{\rm string}(D) \to \Obs^\q_{\rm string}(D')$ is also a quasi-isomorphism. 
\end{proof}

\subsubsection{The stress-energy tensor}

In \cite{WittenTop}, where the notion of a TFT was introduced,
Witten characterized a topological field theory
as a theory whose stress-energy tensor is (homotopy) trivial. 
We now verify that property of the holomorphic bosonic string.
That is, we want to show that the translation symmetries of the holomorphic bosonic string act trivially on the cohomology of the observables.

As a first step, consider the action of the differential operators $\frac{\d}{\d z}$ and $\frac{\d}{\d \zbar}$ on the Dolbeault complex $\Omega^{0,*}(\CC)$. 
This action extends to an action on the fields of the holomorphic bosonic string, and hence to their classical observables as well. 
By Noether's theorem any symmetry of a theory determines classical observables: 
for these symmetries, these are simply the $zz$ and $\zbar \zbar$ components of the stress-energy tensor $T_{zz}$, $T_{\zbar \zbar}$. 
In the case of the bosonic string, we will now show that the stress-energy tensor is cohomologically trivial on the quantum observables.
(Similar but simpler arguments apply to the classical case.)

For each open $U \subset \CC$, the differential operators lift to cochain maps on the quantum level
\ben
\frac{\d}{\d z} , \frac{\d}{\d \zbar} : \Obs^\q_{\rm string} (U) \to \Obs^\q_{\rm string}(U) 
\een 
because the BV Laplacian is translation-invariant.
These cochain maps intertwine with the structure maps of the factorization algebra 
in the sense that they define {\em derivations} of the factorization algebra. 
(See Definition 7.3.2 of \cite{CG1} for a discussion of this notion.)
Note that these operators preserve the cohomological degree. 

Consider now the operator 
\ben
\Bar{\eta} = \frac{\partial}{\partial (\d \zbar)} 
\een 
acting on Dolbeault forms. 
This operator $\Bar{\eta}$ extends to a derivation of degree $-1$ on the factorization algebra $\Obs^\q_{\rm string}$. 
It satisfies the relation
\be\label{d/dzbar}
[\dbar + \hbar \Delta + \{I^\q, -\} , \Bar{\eta}] =  \frac{\d}{\d \zbar}
\ee
as endomorphisms of the factorization algebra, as we now explain.
One observes first that $[\dbar, \Bar{\eta}] =  \frac{\d}{\d \zbar}$. 
Moreover, since $I^\q$ is a chiral deformation, we also have $\Bar{\eta} \cdot I^\q = 0$. 
Finally, since the pairing defining the $-1$-shifted symplectic structure is holomorphic, 
we see that $\Bar{\eta}$ also commutes with the BV Laplacian $[\Bar{\eta}, \Delta] = 0$. 
Hence we have shown the following, by relation~(\ref{d/dzbar}).

\begin{lem}
The operator $\frac{\d}{\d \zbar}$ acts homotopically trivial on $\Obs^\q_{\rm string}$. 
\end{lem}

This fact ensures that the stress-energy tensor vanishes in the $\zbar\zbar$ direction.

We now turn to~$\d/\d z$.
View this vector field $\frac{\d}{\d z}$ as a constant $c$-field. 
Consider the linear local functional of cohomological degree $-2$:
\ben
O_{\frac{\d}{\d z}}(\beta,\gamma,b,c) = \int \<b, \frac{\d}{\d z}\>,
\een
It only depends on the $b$-field.
Note that for this integral to be nonzero, 
the field $b$ must live in $\Omega^{1,1}(\Sigma , T_\Sigma^{1,0*})$.
(In fact, $b$ must also be compactly supported for the integral to be well-defined.) 
Using the BV bracket, we obtain a derivation of the factorization algebra 
\ben
\eta = \{O_{\frac{\d}{\d z}}, -\}
\een
of cohomological degree~$-1$. 
It might help to draw this bracket diagrammatically, 
so one can see that it is a derivation that acts linearly on the generators
(i.e., linear functionals on the fields).

\begin{lem} 
The derivation $\eta$ satisfies 
\be\label{d/dz}
[\dbar + \hbar \Delta + \{I^\q, -\}, \eta] = \frac{\d}{\d z}.
\ee
\end{lem}
\begin{proof}
The derivation $\eta$ commutes with both $\dbar$ and $\Delta$. 
Thus, the left-hand side reduces~to
\[
[\{I, -\}, \eta] = \{ \{I, O_{\frac{\d}{\d z}}\}, -\}.
\]
The only part of the interaction that contributes is $\int \<\beta, c \cdot \gamma\> + \int \<b, [c,c]\>$, and one computes that
\ben
\{I, O_{\frac{\d}{\d z}}\} = \int \<\beta, \partial_z \gamma\> + \int\<b, [\partial_z, c]\> .
\een
Bracketing with this local functional encodes applying $\frac{\d}{\d z}$ to the inputs, as desired.
(In diagrammatic terms, this feature is almost immediately visible.)
\end{proof}

Together these two lemmas ensure that translations act trivially on the cohomological observables.

\subsubsection{Descent for local operators}

We will now sketch an important consequence of the work above.
As we will see, it gives both an approach to the method of descent 
(expositions of this method, as related to two-dimensional gravity, can be found in \cite{WittenDescent,Dijk})
as well as another explicit description of the $E_2$ algebra associated to the quantum observables of the bosonic string.

The key role here is played by observables that are local,
in the sense discussed in Section~\ref{sec: moduli},
where they appeared in our description of the deformation complex, 
but we revisit now the main idea in a more useful form for our current purposes.
We will focus first on the classical theory, 
where the constructions manifestly make sense,
before discussing what needs to be modified in the quantum setting.

Let $J E_{\rm string}$ denote the $\infty$-jet bundle of the classical fields of the holomorphic bosonic string. 
Concretely, a fiber of $J E_{\rm string}$ at a point $x$ corresponds to all the possible Taylor series at $x$ of fields of the bosonic string.
In consequence, the $\infty$-jet of a $\gamma$ field determines a section of $J E_{\rm string}$, 
as does the $\infty$-jet of any other field in the theory.
This bundle $J E_{\rm string}$ is equipped with a canonical flat connection $\nabla^{jet}$ 
such that horizontal sections are precisely the $\infty$-jets of classical fields;
and so $J E_{\rm string}$ is a left $D$-module on the Riemann surface, 
where $D$ means the sheaf of smooth differential operators.
(See the appendix of \cite{GGLA} for expository background oriented toward the approach here.)

Lagrangians can be expressed naturally in terms of $J E_{\rm string}$, 
as sections of the bundle
\[
\Sym(J E_{\rm string}^\vee) = \bigoplus_{k \geq 0}\Sym^k(J E_{\rm string}^\vee),
\]
where $J E_{\rm string}^\vee$ denotes the appropriate dual vector bundle. 
(Some care is required here because $J E_{\rm string}$ is a pro-finite rank vector bundle.)
To unpack this assertion a little, note that a smooth section $\lambda$ of $J E_{\rm string}^\vee$ can be evaluated on the $\infty$-jet of a field to obtain a smooth function on the Riemann surface;
it thus determines a linear functional of fields with values in functions on the surface.
Similarly, a polynomial in such $\lambda$ determines a nonlinear functional on fields with values in smooth functions.
In other words, it is a Lagrangian.
If we multiply it against a density, then we obtain a Lagrangian density and hence a local functional.

Note that sections of $\Sym(J E_{\rm string}^\vee)$ naturally form a graded-commutative algebra, 
since polynomials can be multiplied.
We denote it by~$\sO(J E_{\rm string})$.
The shifted pairing on fields determines a shifted Poisson bracket on $\sO(J E_{\rm string})$,
which we will denote $\{-,-\}$,
since the construction is parallel to the BV bracket.

We will restrict our attention from hereon to $\Sigma = \CC$,
on which $\d^2 z$ determines a natural volume form.
The classical action functional $S$ thus determines a Lagrangian (simply divide by this volume form) that we will abusively denote $S$ as well.
The operator $\{S,-\}$, known as the BRST operator in physics, is square-zero by construction.
Hence, we obtain a commutative dg algebra 
\ben
\left(\sO(J E_{\rm string}), \{S,-\}\right)
\een
in left $D$-modules,
where a flat connection is inherited by the dual bundle and hence by the symmetric powers.

Elements of $\left(\sO(J E_{\rm string}), \{S,-\}\right)$ are not observables of the classical theory,
since they are just Lagrangians.
There is, however, a natural way to produce observables from Lagrangians.
Essentially, a Lagrangian can be multiplied by a de Rham current;
for instance, evaluating a delta function with a Lagrangian produces a pointlike observable
(or local operator in the terminology of physics).

We choose to encode this idea in the following way.
Consider the de Rham complex $\Omega^*(\CC , \sO(J E_{\rm string}))$ of our $D$-module,
equipped with the total differential $\nabla^{jet} + \{S, -\}$. 
It consists of smooth de Rham forms with values in Lagrangians.
These determine observables supported on closed submanifolds, as follows.
For a closed submanifold $C \subset \CC$, 
fix a tubular neighborhood $N_C$.
Integration along $C$ then determines a map 
\begin{equation}
\label{eqn: int over C}
\int_C : \Omega^*(N_C , \sO(J E_{\rm string})) \to \Obs^{\cl}_{\rm string}(N_C) 
\end{equation}
of cochain complexes.

Many of the most familiar observables arise in this fashion.
For instance, take $C$ to be the point at the origin and consider
\[
F(\gamma,\beta,b,c) = \frac{1}{n!}(\partial_z^n \gamma)\big|_{0},
\]
which returns the coefficient of $z^n$ in the Taylor expansion of $\gamma$ around the origin.
(Note that evaluation at $0$ denotes integration at the origin.)
It is easy here to factor $F$ into a Lagrangian term and a de Rham form term: 
the Lagrangian is the linear functional that returns the function $\frac{1}{n!}\partial_z^n \gamma$, 
and the form is the constant function~1.

As a closely related example, take $C$ to be the unit circle and consider
\[
F(\gamma,\beta,b,c) = \frac{1}{2\pi i}\int_C \gamma(z) \frac{\d z}{z^{n+1}}.
\]
For ``on-shell'' $\gamma$ fields---i.e., when $\gamma \in C^\infty(\CC)$ is holomorphic)---this observable $F$ returns the coefficient of $z^n$ in the Taylor expansion around the origin.
It is easy here to factor $F$ into a Lagrangian term and a de Rham form term: 
the Lagrangian is the linear functional that simply returns the function $\gamma$, 
and the form is $(1/2\pi i)\d z/z^{n+1}$, which is well-defined on any small annulus around~$C$.

These concrete examples exhibit a compelling virtue of this process for producing observables: 
it encompasses the observables typically discussed in physics, 
particularly in conformal field theories.
In our situation it is straightforward to show that the entire vertex algebra $\Gr \cV_{\rm string}$ is realized by the image of the map~(\ref{eqn: int over C}).
(Such an argument is given in Part III of \cite{GGW} for the free $\beta\gamma$ system.)

After this lengthy build-up of notions and notations, 
we now finally turn to describing descent.
Our interpretation is that it is a process for promoting pointlike operators to more general observables.
We will soon apply it to give an explanation for the Gerstenhaber structure on $H^* \cV_{\rm string}$,
identified by Lian-Zuckerman.

\begin{dfn}
A {\em pointlike operator} is an element of $\Sym(J E_{\rm string}^\vee)_0$, 
the fiber at the origin $0 \in \CC$ of the bundle~$\Sym(J E_{\rm string}^\vee)$.
\end{dfn}

Equivalently, it is an element of the algebra $\Sym((J_0E_{\rm string})^*)$ of polynomial functions on the fiber at $0$ of $J E_{\rm string}$.
Since our theory is translation-invariant, 
any point in the plane would serve as well as the origin.
Note that this definition is equivalent to our earlier, heuristic notion.

\begin{construction}[Method of "descent"] 
\label{constr: descent}
Any pointlike operator $O$ {\em descends} to an element 
\[
\Tilde{O} = \Tilde{O}^{0} + \Tilde{O}^1 + \Tilde{O}^2 
\] 
in $\Omega^*(\CC, \sO(JE_{\rm string}))$.
We construct it as follows. 
First, because our theory is translation-invariant,
there is a natural trivialization of the bundle $\Sym(J E_{\rm string}^\vee)$,
and hence there is a canonical element $\Tilde{O}^0 \in C^\infty(\CC, \sO(JE_{\rm string}))$,
given by a constant section whose value at the origin is $O$.
In formulas, we write 
\ben
\Tilde{O}^{0} (z) = \tau_z O,
\een
where $\tau_z: \CC \to \CC$ is the translation sending a point $w$ to $w+z$.
(This operator acts on fields by pullback, and so on the jets of fields as well.)
Using the homotopies $\eta,\Bar{\eta}$, 
we define the 1-form part as
\ben
\Tilde{O}^{1} = \d z \; \Tilde{(\eta O)}^0 + \d \zbar \; \Tilde{(\Bar\eta O)}^0,
\een
where $\eta O$ denotes the image of $O$ under the map $\eta$.
Similarly, the 2-form part is 
\[
\Tilde{O}^2 = \d z \,\d \zbar\, \Tilde{(\eta \Bar{\eta} O)}^0.
\] 
By construction, 
if $O$ is closed for the classical differential $\{S,-\}$,
then 
\[
(\d_{dR} + \{S,-\}) \Tilde{O} = 0,
\]
so the total form $\Tilde{O}$ is closed as well.
\end{construction}

\begin{rmk}
The terminology ``descent" is due to Witten \cite{WittenTop}.
Our construction above is a method of solving what he refers to as the topological descent equations. 
Mathematically, we are performing a zig-zag in the double complex given by the de Rham complex of the flat vector bundle of Lagrangians. 
In the horizontal direction, the Lagrangians are equipped with the BRST operator $\{S,-\}$.
In the vertical direction there is the differential induced from the flat connection. 
\end{rmk}

Combining the construction with the map~(\ref{eqn: int over C}), 
we find that a pointlike operator $O$ and a closed submanifold $C$ determine an observable
\ben
\int_{C} \Tilde{O} \in \Obs^\cl_{\rm string} (N_C) .
\een
Note that if $O$ has cohomological degree $k$ and $C$ is of dimension $l$, 
then $\int_C \Tilde{O}$ has degree~$k - l$. 
We remark that  every element of $\Gr \cV_{\rm string}$ can be realized by applying descent to the pointlike operators and then evaluating at the origin.

Extending this whole package to quantum observables is nontrivial. 
The map~(\ref{eqn: int over C}) makes sense at the level of graded vector spaces,
but it is not easy to equip $\sO(J E_{\rm string})$ with a BV Laplacian 
in such a way that the map~(\ref{eqn: int over C}) intertwines with the differential on the quantum observables.
For linear observables, however, no such issues arise
({\em cf.} Part III of \cite{GGW}), 
and those are sufficient to identify the Gerstenhaber bracket,
the problem to which we now turn.

\subsubsection{Formula for the Gerstenhaber bracket}

A Gerstenhaber algebra is a graded commutative algebra with a Lie bracket of cohomological degree $-1$ that is a graded biderivation for the product. 
In this section we show how to explicitly write down the product and bracket on the local observables (i.e., the observables on any disk) and compare our answer to the work of Lian-Zuckerman~\cite{LZ1}.

As explained just before Proposition~\ref{prop: obs is e2}, 
the Gerstenhaber operad $Gerst$ is the operad $H_*(E_2)$ arising by taking homology of the $E_2$ operad.
Recall that $E_2(2)$ parametrizes the space of binary operations as the configuration space of disjoint two disks in the unit disk in $\RR^2$.
This space deformation retracts onto $S^1$.
Hence 
\[
Gerst(2) = H_*(E_2(2)) \cong H_*(S^1). 
\]
(Note that we view homology of spaces as concentrated in nonpositive degrees,
since it is viewed as the linear dual to cohomology.)
The degree zero operation---corresponding to a commutative product---matches with a zero-dimensional cycle of $S^1$,
and the degree -1 operation---corresponding to the shifted Poisson bracket---matches with a one-dimensional cycle of~$S^1$. 

Thus, to obtain the commutative product on $H^*\Obs^\q$, 
we need only pick an embedding of two disjoint disks inside a larger disk,
which is precisely such a zero-cycle in $E_2(2)$.
Then the factorization product
\ben
\Obs^\q(D) \tensor \Obs^\q(D') \to \Obs^\q(D'') .
\een 
induces the commutative product
\[
\cdot: H^*\Obs^\q \otimes H^*\Obs^\q \to H^*\Obs^\q.
\]
Since this configuration space $E_2(2)$ is connected, 
we could use any other choice of embeddings and get the same answer at the level of cohomology.
In particular, we could have put $D'$ on the opposite side of $D$,
which is why the product must be commutative.
(A topologist would call this the Eckmann-Hilton argument,
as it is the same argument one uses to show that the homotopy group $\pi_2(X)$ is always abelian.)

To construct the shifted Poisson (i.e., Gerstenhaber) bracket, 
we need to pick a one-cycle in the configuration space $E_2(2)$.
To describe the associated binary operation, 
we use descent along this one-cycle
in conjunction with the fact---s remarked just after Construction~\ref{constr: descent}---that
the underlying vector space of the vertex algebra $\cV_{\rm string}$ is generated by pointlike operators.

Let $O$ and $O'$ be two pointlike operators supported at $0 \in \CC$. 
Fix a disk $D$ centered at zero.
We can view $O$ as an element in~$\Obs^\q(D)$.

Now fix a loop $C$ that wraps around $D$ but is disjoint.
Let $N_{C}$ be a tubular neighborhood of $C$ that does not intersect $D$. 
Via descent we have the observable $\int_C \Tilde{O}'$ in~$\Obs^\q(N_C)$. 
Pick a bigger disk $D'$ containing $D$ and $N_C$.
Then $\int_C \Tilde{O}'$ also determines an observable in~$\Obs^\q(D')$. 

We suppose that these observables are cocycles, 
which lets us identify them with elements $[O]$ and $[O']$ of the vertex algebra.

Consider now the factorization product
\[
\mu : \Obs^\q(D) \tensor \Obs^\q (N_{C}) \to \Obs^\q(D')
\] 
corresponding to $D \sqcup N_C \hookrightarrow D'$. 
We define a bracket by 
\ben
\{[O],[O']\}_{\rm Ger} := \mu \left(O , \int_{C} \Tilde{O}' \right) .
\een 
Note that if ${\rm deg}(O) = k$ and ${\rm deg}(O') = k'$, then ${\rm deg}(\{[O],[O']\}_{\rm Ger}) = k+k' -1$, 
so we obtain a bracket of the correct degree to define a Gerstenhaber structure. 

\begin{rmk}
This construction manifestly involves picking a 1-cycle, here $C$, to exhibit the bracket, 
and it should be clear geometrically how we could relate to any other choice~$C'$.
If $C$ and $C'$ do not intersect, they bound an annulus and hence determine cohomologous observables.
(One may have to shrink $D$ in the construction, but that is no issue by local constancy.)
If they do intersect, one can choose a $C''$ that does not intersect either, and then one has a pair of cohomologous terms.
As the terms are cohomologous, they induce the same brackets at the level of cohomology.
\end{rmk}

We now connect our constructions with well-known approaches.

\begin{prop} 
The bracket $\{-,-\}_{\rm Ger}$ together with the product $\cdot$ determine the structure of a Gerstenhaber algebra on $H^* \cV_{\rm string}$, 
the cohomology of the dg vertex algebra $\cV_{\rm string}$. 
This Gerstenhaber structure is isomorphic to the one found by Lian-Zuckerman~\cite{LZ1}.
\end{prop}

\begin{proof}
The vertex algebra construction of \cite{CG1} extracts $\cV_{\rm string}$ as the direct sum of the weight spaces of $\Obs^\q_{\rm string}(D)$, 
where $D$ is a disk centered at the origin and we take weight space for the rotation action of $S^1$ on $\CC$.
The bracket and product restrict to this subspace of $\Obs^\q(D)$,
manifestly playing nicely with this eigenspace decomposition. 
Hence they descend to the cohomology of~$\cV_{\rm string}$.

Let $V_{LZ}$ be the Gerstenhaber algebra considered by Lian-Zuckerman.
As vector spaces, both $H^* \cV_{\rm string}$ and $V_{LZ}$ are isomorphic to the state space of the $\beta\gamma$ vertex algebra.

According to the construction of a vertex algebra from a holomorphic factorization algebra in Chapter 6 of \cite{CG1}, the factorization product of two disks is what defines the operator product map $Y(-,z) : V \tensor V \to V ((z))$ of a vertex algebra.
It is this operator product that Lian-Zuckerman use to define the commutative product.
Thus, as commutative algebras, the algebras coincide. 

The brackets coincide by noting that the derivation $\eta$ trivializing $\d / \d z$ agrees with Lian-Zuckerman's trivialization.
\end{proof}

\section{The holomorphic string on closed Riemann surfaces} 
\label{sec: conformalblock}

Thus far we have discussed the local behavior of the holomorphic string,
such as its quantization on a disk and the concomitant vertex algebra.
Now we turn to its global behavior, 
particularly the observables on a closed Riemann surface,
and the relationship with certain natural holomorphic vector bundles on the moduli space of Riemann surfaces.
This local-to-global transition is where the BV/factorization package really shines.
On the one hand, the theory of factorization algebras provides a conceptual characterization of the local-to-global relationship,
much like the understanding of sheaf cohomology as the derived functor of global sections.
On the other hand, the examples from BV quantization provide computable, convenient models for the global sections,
much as the de Rham or Dolbeault complexes do for the cohomology of sheaves that arise naturally in differential or complex geometry.

As we will explain, the answers we recover for the holomorphic string can be related quite cleanly to natural determinant lines on the moduli of Riemann surfaces,
hence providing a bridge from the Feynman diagrammatic anomaly computations to the index-theoretic computations.

\subsection{The free case}

Before jumping to the holomorphic string, 
we will work out the global observables in the simpler case of the $bc\beta\gamma$ system,
introduced in Remark \ref{rmk:bcbg}. 
The global {\it classical}\/ observables on a Riemann surface $\Sigma$ are given by the symmetric algebra on the continuous linear dual to the fields,
\[
\Sym\left(\Omega^{0,*}(\Sigma,V)^\vee \oplus \Omega^{1,*}(\Sigma,V^\vee)^\vee \oplus \Omega^{0,*}(\Sigma,T[1])^\vee \oplus \Omega^{1,*}(\Sigma,T^*_\Sigma[-2])^\vee \right),
\]
with the differential $\dbar$ extended as a derivation.
Hence the cohomology is
\[
\Sym\left(H^*(\Sigma,V)^\vee \oplus H^*(\Sigma,\omega \otimes V^\vee)^\vee \oplus H^*(\Sigma,T[1])^\vee \oplus H^*(\Sigma,\omega^{\otimes 2}[-2])^\vee\right),
\]
where $\omega$ denotes the canonical bundle.
Although this expression might look complicated, 
it can be readily unpacked in the setting of algebraic geometry, 
particularly when $\Sigma$ is closed.
In that case, this graded commutative algebra is a symmetric algebra on a finite-dimensional graded vector space,
which encodes the derived tangent space of the moduli of Riemann surfaces at $\Sigma$ and of holomorphic functions to~$V$.

As this theory is free, it admits a canonical BV quantization.
Denote by $\Obs^{\q}_{free}$ be the corresponding factorization algebra.
One can compute its global sections on $\Sigma$ by using a spectral sequence whose first page is the global classical observables.
The result of Theorem 8.1.4.1 of \cite{CG1} states that the cohomology of this free theory along a closed Riemann surface with values in {\em any} line bundle is one-dimensional concentrated in a certain cohomological degree. 
In our case, it the calculation implies that we get a shifted determinant of the cohomology of the fields:
\ben
H^*\left(\Obs^\q_{free}(\Sigma)\right) \cong \det \left(H^*(\Sigma ; \sO_\Sigma) \right)^{\tensor \dim(V)} \tensor \det \left(H^*(\Sigma ; T_\Sigma^{1,0})\right)^{-1} [d(\Sigma)] 
\een
where 
\ben
d(\Sigma) = \dim (V)  \left(\dim H^0(\Sigma ; \sO_\Sigma) + \dim H^1(\Sigma ; \sO_\Sigma)\right) + \dim(H^0(\Sigma ; T_\Sigma^{1,0})) - \dim(H^1(\Sigma ; T_\Sigma^{1,0})).
\een

\begin{rmk}
The shift $d(\Sigma)$ here likely looks funny.
In this case at least, the meaning can be unpacked pretty straightforwardly. 
The BV complex for an ordinary finite-dimensional vector space is equivalent to the de Rham complex shifted down by the dimension of the vector space, 
so that the top forms are in degree 0.
(Abstracting this situation is one way to ``invent'' the BV formalism.)
For the $\sigma$-model, the global solutions to the equations of motion are $H^0(\Sigma,\sO) \otimes V$ for the $\gamma$ fields and $H^0(\Sigma,\omega) \otimes V^\vee$ for the $\beta$ fields.
For $\Sigma$ closed, these are finite-dimensional, and thus we get the shift
\[
 \dim (V)  \left(\dim H^0(\Sigma ; \sO_\Sigma) + \dim H^1(\Sigma ; \sO_\Sigma)\right).
\]
For the ghost system (the $bc$ fields), 
the BV complex recovers the Euler characteristic 
\[
\dim(H^0(\Sigma ; T_\Sigma^{1,0})) - \dim(H^1(\Sigma ; T_\Sigma^{1,0}))
\]
as it encodes the de Rham complex on the formal quotient stack $B\fg = \ast/\fg$ for the Lie algebra of symmetries~$\fg$.
\end{rmk}

The computation here works for any Riemann surface $\Sigma$ and, indeed, for any family of Riemann surfaces.
Hence it implies that the global observables of the free $bc\beta\gamma$ system determine a determinant line bundle on the moduli $\cM$ of Riemann surfaces.

We can work out the first Chern class of this determinant line bundle using the Grothendieck-Riemann-Roch (GRR) theorem as follows.
Consider the universal Riemann surface $\pi \colon C \to \cM$ over the moduli space, 
and consider the bundles $\sO_C \otimes V$ and $\sT_\pi = \sT_{C/\cM}$,
which one can view the universal $\gamma$ fields and $c$ fields, respectively.
The first Chern class of the derived pushforward $R\pi_*(\sO_C \otimes V)$ is given by the first Chern class of $\det(H^*(\sO_C \otimes V)) \cong \det(\sO_C)^{\otimes \dim V}$, 
since the first Chern class of a vector bundle is the first Chern class of its determinant bundle.
The Grothendieck-Riemann-Roch theorem states that for a complex of coherent sheafs $\cF = \cF^*$ on $C$, 
the Chern character $\ch(R\pi_* \cF)$ of its derived pushforward $R\pi_* \cF$  is given by 
\def\Td{{\rm Td}}
\begin{align*}
\pi_*( \ch(&\cF) \Td(T_\pi)) = \pi_*\left( \left(\sum_{i} (-1)^i \ch (\cF^i) \right) \Td(T_\pi)\right) \\
&= \pi_*\left( \left(\sum_i (-1)^i ({\rm rk}(\cF^i) + c_1(\cF^i) + \frac{1}{2}(c_1(\cF^i)^2) + \cdots)\right) (1 +\frac{1}{2}c_1(T_\pi) + \frac{1}{12} c_1(T_\pi)^2)\right)
\end{align*}
where $T_\pi$ denotes the relative tangent bundle along $\pi$,
which is here just the tangent line bundle of a Riemann surface.
The first Chern class is the component of cohomological degree 2.
For instance, when $\cF = \cF^0$ is concentrated in degree zero, the above simplifies to:
\ben
\frac{1}{12} {\rm rk}(\cF) c_1(T_\pi)^2 + \frac{1}{2} c_1(\cF) c_1(T_\pi) + \frac{1}{2} c_1(\cF)^2 .
\een  
When $\cF = T^{\tensor n}_{\pi} [1]$, the expression for the first Chern class is $-\frac{1 + 6n + 6n^2}{12} c_1(T_\pi)^2$.
When $\cF = \cO \tensor V$ we simply get $\dim(V)$. 

Hence, when $\cF = T[1] \oplus \cO \tensor V$, for the determinant line of global observables $H^*\left(\Obs^\q_{free}(C)\right)$ as a bundle over $C$ we obtain
\[
c_1\left(H^*\left(\Obs^\q_{free}(\Sigma)\right)\right) = \frac{1}{12} (\dim(V) - 13) c_1(T_\pi)^2 .
\]

It is worthwhile to point out that the above argument based on GRR for identifying the first Chern class of this determinant line bundle resonates with our computation of the anomaly of the bosonic string on the disk. Indeed, this is a manifestation of ``Virasoro uniformization.'' 
Also, notice that the above calculation assumed that there was no deformation, so that we were working with a free theory. 
However, deforming the action from free $bc\beta\gamma$ system to holomorphic bosonic string doesn't affect the line bundles, since those are continuous parameters and Chern classes are discrete.

\subsection{The anomaly and moduli of quantizations on an arbitrary Riemann surface}

We have already seen that the holomorphic string {\it on a disk} admits a BV quantization if and only if the target is a complex vector space of dimension 13.
Here we will explain why this anomaly calculation is actually enough to show the existence of a quantization on an {\it arbitrary} Riemann surface. 
An argument using the Grothendieck-Riemann-Roch theorem was given in the above section. 
In this section we give a proof using only the perspective of BV quantization.
One can view this as giving a proof of the Grothendieck-Riemann-Roch theorem using Feynman diagrams (and will be the topic of future work). 

Our diagrammatic arguments showed that only wheels with $c$ legs appear in the anomaly,
and these arguments did not depend on the choice of $\Sigma$. 
Hence the anomaly will be purely a functional on the $c$ fields.
So we restrict ourselves to the piece of the deformation complex only involving such fields.

When $\Sigma$ is a disk,
a corollary of the calculations in Section~\ref{sec: moduli} is that this deformation complex for the $c$-fields is quasi-isomorphic to $\cred^*(W_1)[2]$.
A classical calculation of Gelfand and Fuks shows that the reduced Lie algebra cohomology $H_{red}^*(W_1)$ of formal vector fields is one-dimensional concentrated in degree $3$.
Thus, the cohomology of the deformation complex is 
\[
H^*_{\rm red}(W_1)[2] = \CC[-1].
\]
The generator of this cohomology can be taken to be $\lambda_{-1} \lambda_0 \lambda_1$, where $\lambda_i$ is dual to the formal vector field $L_i = z^{i+1} \partial_z$. 
Some readers might recognize this cocycle as a manifestation of the usual cocycle for the Virasoro Lie algebra.

As shown in Proposition 5.3 of \cite{BWvir}, this situation generalizes from disks to arbitrary Riemann surfaces:
on any Riemann surface $\Sigma$, the deformation complex is equivalent to the derived global sections of $\CC_\Sigma[-1]$, the constant sheaf on $\Sigma$ with a cohomological shift.
In particular, the cohomology of the deformation complex is equal to the shift of the de Rham cohomology $H^*(\Sigma)[-1]$. 
Hence, the first cohomology group of the deformations---where the anomaly lives---is spanned by the constant functions.
When $\Sigma$ is connected, we thus know the anomaly, up to scale.
Locally, this anomaly cocycle agrees with the usual expression for the generator of $H^3(W_1)$, 
but since this description depends on the choice of a coordinate, 
the global version is somewhat subtle. 
See Section 5 of \cite{BWvir} for an extensive discussion. 

The sheaf-theoretic nature of the deformation complex works to our advantage here.
As the construction of BV quantization is manifestly {\em local-to-global} on spacetime 
anomalies inherit this property: the anomaly computed on an open set $U \subset \Sigma$ is equal to the anomaly of the theory on $\Sigma$ restricted to $U$. 
In our case, this global anomaly is a 1-cocycle for the derived global sections of the shifted constant sheaf, 
and hence is determined by a constant function on $\Sigma$.
Thus, it suffices to compute the anomaly on an arbitrary open, 
such as on a flat disk. 
But this is precisely the context in which we computed the anomaly in Section~\ref{sec: quantization}, 
so we know the anomaly is simply the dimension of the target vector space.
Thus, a quantization of the holomorphic string exists on any Riemann surface provided~$\dim_{\CC}(V) = 13$. 

Now we ask how many such quantizations are possible,
i.e., what is the moduli of theories.
By the calculation in Section \ref{sec: moduli}, 
we know that, up to BV equivalence, 
the possible one-loop terms in the quantized action functional are parametrized by
\ben
H^0(\Sigma) \tensor \Omega^1(V) \oplus H^1(\Sigma) \tensor \Omega^2_{cl}(V).
\een 
(That is, these vector spaces are the first cohomology group of the relevant deformation complex.)
This space of deformations corresponds to continuous parameters we can vary in the action functional.
As the isomorphism classes of line bundles form a discrete set, 
varying these continuous parameters will not change the class of the line bundle of global observables. 
In conclusion, no matter what one-loop quantization we choose, 
the cohomology of the global observables will be the same.

\subsection{Global observables for the holomorphic string}

Now, let us consider the global observables of the holomorphic string $\Obs^\q(\Sigma)$. 
Consider the filtration on the quantum observables induced by the polynomial degree of the functional. 
There is a spectral sequence abutting to the cohomology of the global observables $H^*\Obs^\q(\Sigma)$ with $E_1$ page given by the cohomology of the global observables of the free $bc\beta \gamma$ system which we have already computed:
\bestar
E_2 & \cong & \det\left(H^*(\Sigma ; T_\Sigma[1])\right) \tensor \det \left(H^*(\Sigma ; \cO_\Sigma)^{\oplus 13}\right) \\
& \cong & \det \left(H^1(\Sigma ; T_\Sigma) \right) \tensor \det \left(H^0(\Sigma ; T_\Sigma)\right)^{-1} \tensor \det \left(H^0(\Sigma ; K_{\Sigma}) \right)^{-13}
\eestar
where we have used the fact that $H^0(\Sigma ; \cO) \cong \CC$ for any $\Sigma$. 
Since this page is concentrated in a single line, we see that the spectral sequence degenerates at this page.

Let $\Sigma_{g}$ be a surface of genus $g$. Then for $g=1$ the above simplifies to
\ben
\det \left(H^1(\Sigma_1 ; T_{\Sigma_1})\right) \tensor \det \left(H^0(\Sigma_1 ; K) \right)^{-14} .
\een 
If $g \geq 2$ one has
\ben
\det \left(H^1(\Sigma_1 ; T_{\Sigma_1})\right) \tensor \det \left(H^0(\Sigma_1 ; K) \right)^{-13} .
\een
Thus the above expressions give the global observables for the holomorphic string for genus $g =1$ and $g \geq 2$, respectively. 
Compare these formulas to Witten's analysis of the bosonic string in \cite{WitString}.

\section{Looking ahead: curved targets}
\label{sec:curved}

In this section we briefly advertise our future work, which is to provide a complete analysis of the bosonic string with a complex manifold as the target. 
Our approach is a modification of our treatment of the curved $\beta\gamma$ system given in \cite{GGW}.
The main idea there was to consider the $\beta\gamma$ system with target a formal disk $\hD^n$. 
Then, in the style of Gelfand and Kazhdan's treatment of formal geometry, we show how working equivariantly for formal automorphisms allows us to globalize this theory to a complex manifold. 
In general, we find an obstruction to doing this, which is measured by the second component of the Chern character of the tangent bundle of the complex manifold. 
The appearance of the characteristic class is expected from the theory of chiral differential operators.
In fact, we show that the factorization algebra of observables descends to the sheaf of chiral differential operators on the target manifold. 

We will give a similar argument for the bosonic string. 
The key difference to the $\beta\gamma$ system is that even in the case of a flat target, the bosonic string is an interacting theory.
Nevertheless, the theory of BV quantization that is equivariant for formal automorphisms can still be applied and we arrive at the following result. 

\begin{thm}[\cite{GWcurved}] 
Consider the holomorphic bosonic string with target a complex manifold $X$. 
There exists a one-loop exact quantization if and only if
\begin{itemize}
\item[(1)] $\dim_\CC X = 13$,

\item[(2)] $\ch_2(T_X) = 0$, and

\item[(3)] $c_1(T_X) = 0$.
\end{itemize}
Moreover, if the conditions above hold, the space of all quantizations is a torsor for the abelian group
\ben
H^1(X , \Omega^2_{cl}) \oplus H^0(X, \Omega^1) .
\een
\end{thm}

There are two further directions we intend to address in our future work:
\begin{itemize}
\item[(1)] We have seen that the local observables for the case of a flat target return the semi-infinite BRST cohomology of the $\beta\gamma$ vertex algebra. 
We expect that the local observables in the case of a curved target produce a sheafy refinement of semi-infinite cohomology. 
This should produce a sensitive invariant of the target manifold and gives a variant of quantum sheaf cohomology. 
\item[(2)] The partition function of the curved $\beta\gamma$ system on elliptic curves is known to produce the Witten genus of the target manifold \cite{wg2}. 
For flat space, the partition function of the string is given by the Mumford form \cite{BM}. 
We propose that the partition function for the curved string produces an invariant of the target manifold analogous to the Witten genus. 
\end{itemize}


\appendix

\section{Calculation of anomaly} \label{sec:calculation}

In this section we compute the functional $F[L]$ and $G[L]$ mentioned in the proof of Proposition~\ref{prop anomaly}, hence completing the calculation of the anomaly. 

We have reduced the calculation to the weight of two wheel diagrams: A) with internal edges labeled by the $bc$ heat kernel and propagator, respectively. B) with internal edges labeled by the $\beta\gamma$ heat kernel and propagator, respectively.
The weight of A gives the functional we called $F[L]$, and the weight of B gives the functional we called~$\dim_\CC(V) G[L]$. 

We will utilize the following version of Wick expansion to evaluate the integrals below. 

\begin{lem}\label{lem wick} Let $\Phi$ be a smooth compactly supported function on $\CC$ and suppose $\tau > 0$. 
Then
\ben
\int_{\xi \in \CC} \Phi(\xi) e^{-\tau |\xi|^2/4} \d^2 \xi  = 4 \pi \cdot \tau^{-1} \left(\exp\left(\tau^{-1} \frac{\partial}{\partial \xi} \frac{\partial}{\partial \xi} \Phi\right)_{\xi = 0}\right) .
\een
\end{lem}

We start with the weight of diagram~A. 
Use coordinates $z,w$ to denote the coordinates at each of the vertices.
Denote the inputs of the weight by the compactly supported vector fields $f(z) \partial_z$ and $g(w) \d \wbar \partial_w$.
(Note that the diagram is only nonzero if the total degree of the elements is $+1$.)
If $c(z) \partial_z$ is another vector field, the action by $f(z) \partial_z$ is given by 
\ben
[f(z) \partial_z, c(z) \partial_z] = f(z) \partial_z c(z) \partial_z - c(z) \partial_z f(z) \partial_z .
\een 
Thus, the weight of diagram $A$ can be written as the $\ell \to 0$ limit of
\be
\begin{array}{ccc}
\displaystyle
& & \int_{z,w} f(z) \partial_z P_{\ell}^L(z,w) g(w) \partial_w K_\ell(z,w) \\
&-& \int_{z,w} \partial_z f(z) P_{\ell}^L(z,w) g(w) \partial_w K_\ell (z,w) \\
&-& \int_{z,w} f(z) \partial_z P_\ell^L(z,w) \partial_w g(w) K_\ell (z,w) \\
&+& \int_{z,w} \partial_z f(z) P_\ell^L(z,w) \partial_w g(w) K_\ell (z,w) .
\end{array}
\ee
We label the integrals in each line above as I,II, III, IV, respectively. 

Using the form of the propagator in (\ref{propagator}) we see that line I is given by
\ben
{\rm I} = \frac{1}{(4 \pi)^2} \int_{(z,w) \in \CC \times \CC} \int_{t = \ell}^L f(z) g(w) \frac{1}{\epsilon^2} \frac{1}{t^3} \frac{(\zbar - \wbar)^3}{8} \exp \left(-\frac{1}{4}\left(\frac{1}{\ell} + \frac{1}{t}\right)|z-w|^2 \right)
\een
(we are omitting volume factors for simplicity). 
To evaluate this integral we change variables and apply the Wick expansion, 
Lemma~\ref{lem wick}, to one of the variables of integration. 
Indeed, introduce $\xi = z -w$, and notice that the integral simplifies to
\ben
{\rm I} = \frac{1}{(4 \pi)^2} \int_{w \in \CC} \int_{\xi \in \CC} \int_{t = \ell}^L f(\xi + w) g(w) \frac{1}{\epsilon^2} \frac{1}{t^3} \frac{\Bar{\xi}^3}{8} \exp \exp \left(-\frac{1}{4} \left(\frac{1}{\ell} + \frac{1}{t}\right)| \xi |^2 \right) .
\een
Applying Lemma \ref{lem wick} to the $\xi$-integral we see that this simplifies to
\ben
{\rm I} = \frac{1}{4 \pi} \int_{w \in \CC} \partial^3_w f(w) g(w) \int_{t = \ell}^L \frac{\ell^2 t}{(\ell + t)^4} + O(\ell)
\een
where the terms $O(\ell)$ are of order $\ell$ so are zero in the limit $\ell \to 0$. 
On the other hand, we can evaluate the remaining $t$-integral and see that in the limit $\ell \to 0$,
Line I becomes
\ben
\lim_{\ell \to 0} \; {\rm I} = \frac{1}{4 \pi} \frac{1}{12} \int_{w\in \CC} \partial^3_w f(w) g(w) \d^2 w .
\een 
We evaluate II, III, and IV in a similar fashion.

After changing coordinates and performing the Wick type integral, 
we obtain
\ben
{\rm II} = \frac{1}{4 \pi} \int_{w \in \CC} \partial^3_w f(w) g(w) \d^2 w \int_{t = \ell}^L \frac{\ell t}{(\ell + t)^3} \d t+ O(\ell) .
\een
Evaluating the remaining $t$ integral and taking $\ell \to 0$, 
this becomes 
\ben
\lim_{\ell \to 0} {\rm II} = \frac{1}{4 \pi} \frac{3}{8} \int_{w\in \CC} \partial^3_w f(w) g(w) \d^2 w .
\een 

Integral III is given by 
\ben
\frac{1}{4\pi} \int_{w \in \CC} \partial_w^3 f(w) g(w) \d^2 w \int_{t= \ell}^L \frac{\epsilon^2}{(\epsilon + t)^3} \d t + O(\ell) .
\een
In the limit $\ell \to 0$ we obtain
\ben
\lim_{\ell \to 0} \; {\rm III} = \frac{1}{4\pi} \frac{1}{8} \int_{w \in \CC}  \partial^3 f(w) g(w) \d^2 w .
\een

Finally, integral IV is
\ben
\frac{1}{4\pi} \int_{w \in \CC} \partial_w^3 f(w) g(w) \d^2 w \int_{t= \ell}^L \frac{\epsilon}{(\epsilon + t)^2} \d t + O(\ell) .
\een
In the limit $\ell \to 0$, we obtain
\ben
\lim_{\ell \to 0} \; {\rm IV} = \frac{1}{4\pi} \frac{1}{2} \int_{w \in \CC}  \partial^3_w f(w) g(w) \d^2 w .
\een

In total, the functional $F[L]$ applied to $(f(z) \partial_z, g(w) \d \wbar \partial_w)$ is given by
\ben
F[L] (f(z) \partial_z, g(w) \d \wbar \partial_w) = - \frac{1}{4 \pi} \frac{13}{12} \int_{w \in \CC}  \partial^3_w f(w) g(w) \d^2 w .
\een
Note that this functional is independent of~$L$. 

Diagram B is similar to A, except the internal edges are labeled by the $\beta\gamma$ propagator. 
Applied to the input vector fields $(f(z) \partial_z, g(w) \d \wbar \partial_w)$ the weight is given by the dimension of $V$ times the integral we computed in~$I$.
Thus
\ben
G[L](f(z) \partial_z, g(w) \d \wbar \partial_w) = \frac{1}{4 \pi} \frac{1}{12} \int_{w\in \CC} \partial^3_w f(w) g(w) \d^2 w .
\een
The proposition follows.

\bibliographystyle{amsalpha}

\bibliography{string}

\end{document}